\documentclass[11pt,a4paper,reqno]{amsart}
\usepackage{amsthm,amsmath,amsfonts,amssymb,amsxtra,appendix,bookmark,euscript,delarray,dsfont,mathrsfs}
\usepackage{enumerate}
\usepackage{amssymb}
\usepackage[utf8]{inputenc}
\usepackage{xcolor}
\usepackage{graphicx}
\usepackage{upgreek}
\usepackage[normalem]{ulem}
\usepackage{ stmaryrd }

\allowdisplaybreaks

\usepackage{tikz}
\usetikzlibrary{decorations.pathreplacing, calc}

\theoremstyle{plain}
\newtheorem{theorem}{Theorem}
\newtheorem{lemma}{Lemma}
\newtheorem*{lemma*}{Lemma}
\newtheorem{corollary}{Corollary}
\newtheorem{proposition}{Proposition}[section]

\newtheorem*{conjecture*}{Conjecture}
\newtheorem{assumption}{Assumption}
\theoremstyle{definition}

\theoremstyle{remark}
\newtheorem{remark}{Remark}

\numberwithin{equation}{section}

\renewcommand{\d}{\mathrm{d}}                       

\newcommand{\dd}{\,\mathrm{d}}    
\newcommand{\ii}{\mathrm{i}}
\newcommand{\comm}[2]{\left[#1,#2\right]}

\DeclareMathOperator{\Ker}{Ker}

\DeclareMathOperator{\tr}{Tr}

\renewcommand{\leq}{\leqslant}
\renewcommand{\geq}{\geqslant}
\renewcommand{\ge}{\geqslant}
\renewcommand{\le}{\leqslant}

\def\eps {\varepsilon}

\renewcommand{\to}{\rightarrow}
\renewcommand{\phi}{\varphi}

\newcommand{\norm}[1]{\left\lVert #1 \right\rVert}

\renewcommand{\phi}{\varphi}
\newcommand{\ket}[1]{|#1 \rangle}       
\newcommand{\bra}[1]{\langle #1 |}      
\newcommand		{\lt}			{\left}				
\newcommand		{\rt}			{\right}

\newcommand		{\bangle}[1]	{\lt\langle #1\rt\rangle}
	
\newcommand		{\inprod}[2]	{\bangle{#1, #2}}

\newcommand\R{{\ensuremath {\mathbb R} }}
\newcommand{\CC}{{\ensuremath {\mathbb C} }}
\newcommand\N{{\ensuremath {\mathbb N} }}
\newcommand\Z{{\ensuremath {\mathbb Z} }}

\newcommand\1{{\ensuremath {\mathds 1} }}
\newcommand{\gH}{\mathfrak{H}}
\newcommand{\h}{\hbar}

\newcommand{\gS}{\mathfrak{S}}

\newcommand{\cD}{\mathcal{D}}

\newcommand{\cL}{\mathcal{L}}

\newcommand{\cI}{\mathcal{I}}

\newcommand{\bT}{\mathbb{T}}

\usepackage{tcolorbox}
\newtcolorbox{mybox}{colback=white,colframe=orange,boxrule=1.5pt,arc=3mm}

\title[Quantum Relative Entropy]{Quantum Relative Entropy and the Mean-Field Limit}

\author[G. Guo]{Gaoyue Guo}
\address{Université Paris-Saclay CentraleSupélec, Laboratoire MICS and CNRS FR-3487}
\email{gaoyue.guo@centralesupelec.fr}

\author[H. Liang]{Hao Liang}
\address{School of Mathematical Sciences, Peking University, Beijing, 100871, China}
\email{leunghao@stu.pku.edu.cn}

\author[Z. Wang]{Zhenfu Wang}
\address{Beijing International Center for Mathematical Research, Peking University, Beijing, 100871, China}
\email{zwang@bicmr.pku.edu.cn}

\begin{document}

\begin{abstract} 
We develop a quantum relative entropy method for the mean-field limit of quantum many-body systems. For closed systems governed by the von Neumann equation, we prove a quantitative stability estimate between the $N$-body density matrix and the tensorized solution of the Hartree equation. The argument is based on an entropy production identity, a cancellation mechanism for the centered two-body fluctuation, and a combinatorial estimate controlling the remaining mixed moments. As a consequence, we obtain propagation of chaos in trace norm for fixed marginals. We further combine the entropy estimate with known semiclassical Wasserstein bounds to derive a convergence estimate that is uniform in the Planck constant in an appropriate joint mean-field and semiclassical regime. Finally, we extend the method to finite-dimensional open quantum systems governed by Lindblad dynamics. In this setting, we establish an analogous relative entropy estimate for general bounded two-body interactions, where the mean-field potential is defined through partial trace. This shows that the entropy method does not rely on any special tensor-product decomposition of the interaction.

\end{abstract}



\maketitle

\tableofcontents

\section{Introduction}

\subsection{Continuous $N$-Body Systems and Mean-Field Limits}

Consider a system of $N$ identical particles of mass $m$ confined within a $d$-dimensional physical domain $\Omega$. In particular, we consider two primary geometric settings: the flat torus of side length $L$, denoted by $\mathbb{T}_L^d\cong (\R/L\mathbb{Z})^d$, or the entire Euclidean space $\R^d$. In the case of the torus, $L$ naturally prescribes the periodicity and the volume of the configuration space, whereas in the Euclidean case, $L$ serves as a characteristic observation scale.

The evolution of the $N$-particle wave function $\Psi(t,x_1,\ldots,x_N)$ is governed by the linear Schrödinger equation:
\begin{equation}\label{eq:Schrodinger}
    \ii \hbar\partial_t \Psi=-\frac{\hbar^2}{2m}\sum_{j=1}^N\Delta_{x_j}\Psi+\sum_{1\leq j<k\leq N}V(x_j-x_k)\Psi,
\end{equation}
where $V$ is a real-valued, even potential function describing the interaction between particles and $\hbar$ is the reduced Planck constant. For $\Omega=\mathbb{T}_L^d$, the potential $V$ is assumed to be $L$-periodic.

To analyze the collective dynamics at macroscopic scales, we introduce a rescaling of the space-time coordinates. Let $\hat{x}_j=x_j/L$ be the dimensionless coordinates such that the new spatial domain is the unit torus $\mathbb{T}^d$ or $\R^d$, respectively. We also rescale time with a characteristic scale $T$ such that $t=T\hat{t}$.

We define a dimensionless parameter $\hat{\hbar}$, which we still call the Planck constant below, by $\hat{\hbar}:=\frac{\hbar T}{2mL^2}$, and rescale the interaction potential as follows:
\begin{equation*}
    \hat{V}(\hat{x}):=\frac{(N-1)T\hat{\hbar}}{\hbar}V(L\hat{x}).
\end{equation*}
This choice produces the usual \textit{mean-field scaling}, in which the interaction energy per particle remains of order one as $N\to\infty$.

Upon defining the rescaled wave function 
\begin{equation*}
    \hat{\Psi}(\hat{t},\hat{x}_1,\ldots,\hat{x}_N):=L^{\frac{dN}{2}}\Psi(t,x_1,\ldots,x_N),
\end{equation*}
the Schrödinger equation \eqref{eq:Schrodinger} takes the following dimensionless form:
\begin{equation*}
    \ii \hat{\hbar}\partial_{\hat{t}}\hat{\Psi}=-\hat{\hbar}^2\sum_{j=1}^N\Delta_{\hat{x}_j}\hat{\Psi}+\frac{1}{N-1}\sum_{1\leq j<k\leq N}\hat{V}(\hat{x}_j-\hat{x}_k)\hat{\Psi}.
\end{equation*}
With a slight abuse of notation, we will drop the hat. We study the Hamiltonian $H_{N,\hbar}$ defined by 
\begin{equation*}
    H_{N,\hbar}=\sum_{j=1}^N -\hbar^2\Delta_{x_j}+\frac{1}{N-1}\sum_{1\leq j<k\leq N}V(x_j-x_k),
\end{equation*}
and the time evolution of $N$-body mixed states, which are self-adjoint, positive trace-class operators on $L^2(\Omega^N)$ with trace one. The dynamics of such states is governed by the von Neumann equation (the generalization of the Schrödinger equation to density operators):
\begin{equation}\label{eq:introduction von Neumann}
    \ii \hbar\partial_t\Gamma_t^N=[H_{N,\hbar},\Gamma_t^N],
\end{equation}
which is the quantum analogue of the classical Liouville equation. 

We are mainly interested in the large-$N$ limit (also called the mean-field limit in our scaling) $N\to \infty$, and we also discuss a joint regime for the von Neumann equation in which $N\to \infty$ and $\hbar \to 0$. 

\subsubsection{Quantum mean-field limit: fixed $\hbar$ and large $N$.}

When $N$ is large, the full density matrix $\Gamma_t^N$ contains much more information than is typically observable. We are instead interested in the behavior of a few-body observables, which can be extracted via the \textit{reduced density matrices} (or marginals). For any $1\leq k\leq N$, we define the $k$-th marginal $\Gamma_t^{N:k}$ of $\Gamma_t^N$ by taking the partial trace over the remaining $N-k$ particles:
\begin{equation*}
    \Gamma_t^{N:k}:=\tr_{k+1,\ldots,N}(\Gamma_t^N).
\end{equation*}
In terms of its integral kernel, the first marginal $(k=1)$ is given by 
\begin{equation*}
    \Gamma_t^{N:1}(x,y):=\int_{\Omega^{N-1}}\Gamma_t^N(x,z_2,\cdots,z_N;y,z_2,\cdots,z_N)\dd z_2\cdots\dd z_N.
\end{equation*}

The central objective of mean-field theory is to show that as $N \to \infty$ (while $\hbar>0$ is held fixed), the $N$-particle dynamics decouple. Suppose that the initial data $\Gamma_0^N$ is permutation symmetric and approximately factorized, i.e. $\Gamma_0^N\approx \gamma_0^{\otimes N}$ for some one-particle density operator $\gamma_0$. The goal is to prove that this factorized structure is approximately preserved by the evolution \eqref{eq:introduction von Neumann}, such that 
\begin{equation}\label{eq:introduction convergence marginals}
    \Gamma_t^{N:k}\longrightarrow \ \gamma_t^{\otimes k}\quad {\rm as}\ N\rightarrow \infty,
\end{equation}
where $\gamma_t$ solves the (nonlinear) Hartree equation:
\begin{equation*}
    \ii \hbar\partial_t\gamma_t=[-\hbar^2\Delta+V\ast\rho_{\gamma_t},\gamma_t],
\end{equation*}
where $\rho_{\gamma_t}(x)=\gamma_t(x,x)$ denotes the spatial density of $\gamma_t$. 

To quantify the convergence of the marginals $\Gamma_t^{N:k}$ to the mean-field limit, one first needs to specify an appropriate metric. A natural choice is the \textit{trace norm} (or $\mathcal{L}^1$-norm), which is the quantum analogue of the \textit{total variation norm} for classical probability densities.

In the  mean-field limit  for classical interacting particle systems, Jabin and the third-named author introduced a relative entropy method for systems with bounded interaction kernels in \cite{jabin2016mean}. The method relies on the evolution of the relative entropy together with uniform-in-$N$ bounds for an associated partition function, and was later extended in \cite{jabin2018quantitative} to a wider class of singular interactions. These works provide a natural motivation for the present paper, where we investigate a quantum analogue of this entropic approach.

The quantum relative entropy, defined for two density operators $\Gamma$ and $\Gamma'$ by
\begin{equation}
    S(\Gamma,\Gamma'):=\tr(\Gamma\log\Gamma-\Gamma\log\Gamma'),
\end{equation}
is a natural quantum analogue of the classical relative entropy. It plays a central role in many-body analysis, providing a powerful tool for proving convergence and stability results. In \cite{lewin2021classical}, quantum relative entropy is used to control the deviation of reduced density matrices, weighted by a one-body operator, from those associated with a Gaussian state. This estimate is a key ingredient in the derivation of the quantum-to-classical limit.

We employ the quantum relative entropy to compare the solution of the von Neumann equation, $\Gamma_t^N$, with the $N$-fold tensor product $\gamma_t^{\otimes N}$ of the solution to the Hartree equation. Under suitable regularity assumptions on the initial data $\gamma_0$ and the finiteness of the initial relative entropy, we prove that the relative entropy remains controlled on every finite time interval. More precisely, we have 
\begin{equation*}
 S(\Gamma_t^N,\gamma_t^{\otimes N})\leq C_t \left(S(\Gamma_0^N,\gamma_0^{\otimes N}) +\log 2\right).
\end{equation*}
See Theorem \ref{tm: relative entropy bound}. As a corollary, we quantify the convergence \eqref{eq:introduction convergence marginals} in the trace norm for any fixed $k$ as $N\to \infty$.

\begin{remark}[Symmetry and Particle Statistics]
While the permutation symmetry of the initial state $\Gamma_0^N$ reflects the indistinguishable nature of the particles, we emphasize that the scaling $(N-1)^{-1}$ with fixed $\hbar$ and the resulting Hartree limit are specifically tailored for \textit{bosonic} mean-field theory. For fermionic systems, the Pauli exclusion principle leads to a different scaling of the kinetic energy relative to the particle number (typically $\hbar \sim N^{-1/d}$), necessitating a distinct mathematical treatment and leading to different effective models, such as the Hartree-Fock theory. The mean-field limit for fermions was first discussed by Elgart, Erd{\H{o}}s, Schlein, and Yau \cite{elgart2004nonlinear}. For general mixed-state initial data, see \cite{benedikter2016mean}. 
\end{remark}

For a fixed $\hbar>0$, the mean-field limit is by now well understood, including for singular interactions such as the Coulomb potential. The literature contains results on weak convergence \cite{bardos2002derivation, bardos2000weak, erdos2001derivation} and explicit rates of convergence in suitable norms \cite{chen2018rate, chen2011rate, grillakis2010second, kuz2015rate, mitrouskas2019bogoliubov,paul2019size,pickl2011simple,rodnianski2009quantum,deuchert2023dynamics}. See \cite{elgart2007mean} for a relativistic dispersion result of the mean-field limit.

\bigskip

\subsubsection{Uniformity in $\hbar$}

The joint limit $(N, \hbar) \to (\infty, 0)$ corresponds to the transition from the $N$-body quantum dynamics to the classical Vlasov equation. Since the von Neumann and Hartree equations formally converge to their classical counterparts as $\hbar \to 0$ (see \cite{lions1993mesures} for example), it is expected that $\Gamma_t^{N:1}$ converges to a probability density $f_t(x, v)$ on the phase space satisfying the Vlasov equation:
\begin{equation}\label{eq:introduction Vlasov}
    \partial_t f_t +  v \cdot \nabla_x f_t - \nabla_x (V \ast \rho_{f_t}) \cdot \nabla_v f_t = 0.
\end{equation}
The central challenge in this context is to obtain convergence estimates that are uniform in $\hbar$, or at least valid in a  regime where $\hbar$ depends on $N$.

The first rigorous derivation of the Vlasov equation \eqref{eq:introduction Vlasov} from quantum many-body systems was obtained by Narnhofer and Sewell \cite{narnhofer1981vlasov} in the case of smooth potentials and with $\hbar=N^{-1/3}$. Then, Spohn and Neunzert \cite{spohn1981vlasov} extended the result to the case of twice differentiable potentials. For the same kind of potentials, a more explicit rate of convergence without assuming $\hbar=N^{-1/3}$ was later obtained by Graffi, Martinez, and Pulvirenti \cite{graffi2003mean}, and by Golse and Paul \cite{golse2017schrodinger} in the quantum Wasserstein metrics (see also \cite{golse2016mean} for more about the quantum Wasserstein metrics). More recently, Chong, Lafleche, and Saffirio \cite{chong2024many} derived the Vlasov equation from many-body fermionic systems with mild singular interactions. 

As for uniformity in $\hbar$, Fr{\"o}hlich, Graffi, and Schwarz \cite{frohlich2007mean} established an estimate, uniform in $\hbar$, for an appropriate distance between the $N$-body quantum dynamics and the Hartree dynamics; this applies only to velocity-dependent interactions. Later, Golse, Paul, and Pulvirenti established a mean-field convergence rate of $1/N$ that is uniform in the Planck constant over short time intervals for analytic interaction potentials and initial data, utilizing a self-contained treatment of the BBGKY hierarchy; see Theorem 3.2 in \cite{golse2018derivation}. Alternatively, for more general initial data and interactions with Lipschitz forces, the authors employ the BBGKY hierarchy in conjunction with the quantum Wasserstein distance in \cite{golse2016mean} to establish an $\hbar$-uniform convergence rate of $1/\sqrt{\log\log N}$, see Theorem 3.1 in \cite{golse2018derivation}, which can be compared with our second main result, Theorem \ref{tm: uniform in h}, where we use the entropy method instead of the BBGKY hierarchy and obtain a slightly better convergence rate.

\bigskip

\subsection{The Lattice Framework: Open System Lindbladians.}
Markovian open quantum systems are commonly described by \textit{Lindblad equations}, which generalize the von Neumann equation to include dissipative effects. In his seminal 1976 work \cite{lindblad1976generators}, G. Lindblad established the definitive structure of the generators for completely positive, trace-preserving semigroups on a $*$-algebra (typically the space of bounded operators on a Hilbert space).

The canonical form of the generator---the so-called Lindbladian---is expressed as
\begin{equation*}
    \mathfrak{L}(\rho):=-\ii [H,\rho]+\sum_{j} \left( L_j\rho L^*_j-\frac{1}{2} (L^{*}_j L_j\rho+\rho L_j^* L_j)\right). 
\end{equation*}
A cornerstone of his derivation is the assumption of \textit{norm continuity} for the semigroup, which is equivalent to the requirement that the generator $\mathfrak{L}$ be a bounded linear map. See Theorems 1 and 2 in \cite{lindblad1976generators}. Under this analytical constraint, the Hamiltonian $H$ and the Lindblad operators $L_j$ must be bounded operators. While this framework provides the definitive algebraic structure for Markovian evolution, it poses a significant challenge when considering particles in continuous space, where the Hamiltonian typically involves unbounded operators such as the kinetic energy $-\Delta$.

To investigate the mean-field limit of such systems within a mathematically consistent framework that adheres to Lindblad's original setting, we use a lattice model, such as the \textit{Bose-Hubbard model}, as a guiding example. It is defined on a finite lattice $\Lambda\subset \Z^d$ with periodic boundary conditions. The single-particle Hilbert space is $L^2(\Lambda)\cong \CC^{|\Lambda|}$. The many-body Hamiltonian is given by
\begin{equation*}
    H_N=\sum_{j=1}^N -\Delta^{\mathrm{L}}_j+\frac{1}{N-1}\sum_{1\leq j<k\leq N}\delta_{x_j,x_k},
\end{equation*}
which is an operator defined on the $N$-particle Hilbert space $L^2(\Lambda)^{\otimes N}$. Here $-\Delta_j^{\mathrm{L}}$ denotes the lattice (weighted graph) Laplacian acting on the $j$-th tensor factor. For $u\in L^2(\Lambda)$,
\begin{equation*}
    -\Delta^{\mathrm{L}}u(x):=\sum_{y\in\Lambda,|y-x|=1}(u(x)-u(y)).
\end{equation*}
Let $\{\ket{k}\}_{k\in \Lambda}$ be the standard basis of $L^2(\Lambda)$. Then the on-site interaction can be written as the two-body operator
\[
    W=\sum_{k\in\Lambda}\ket{k}\bra{k}\otimes\ket{k}\bra{k},
\]
and the Hamiltonian becomes
\begin{equation*}
    H_N=\sum_{j=1}^N -\Delta^{\mathrm{L}}_j+\frac{1}{N-1}\sum_{1\leq j<k\leq N}W_{jk}.
\end{equation*}
In the main theorem for open systems, however, we do not need to restrict the two-body interaction to this particular on-site form. We work on a general finite-dimensional Hilbert space $\gH$ and allow an arbitrary bounded self-adjoint two-body interaction $W=W^*\in\cL(\gH\otimes\gH)$, assumed symmetric under exchange of the two tensor factors. The corresponding mean-field potential is then defined intrinsically by partial trace,
\[
    V^\gamma:=\tr_2\big((\1\otimes\gamma)W\big).
\]
This formulation includes the Bose-Hubbard interaction as a special case, but it avoids imposing an unnecessary tensor-product decomposition of $W$.

For a bounded operator $L$, we define the one-body dissipator
\begin{equation*}
    \cL_{L}(A):= LAL^*-\frac{1}{2}\left( L^* LA+AL^* L \right). 
\end{equation*}
We consider the $N$-body Lindblad equation
\begin{equation*}
    \partial_t\Gamma_t^N=-\ii [H_N,\Gamma_t^N]+\sum_{j=1}^N \cL_{L_j}(\Gamma_t^N),
\end{equation*}
where $L_j$ is the operator obtained by letting $L$ act on the $j$-th variable. The corresponding mean-field equation is the Hartree--Lindblad equation given by
\begin{equation*}
    \partial_t\gamma_t
    =
    -\ii[h+V^{\gamma_t},\gamma_t]+\cL_L(\gamma_t).
\end{equation*}

We again use quantum relative entropy to compare $\Gamma_t^N$ with $\gamma_t^{\otimes N}$, and prove an entropy stability estimate; see Theorem \ref{thm: finite main}.

\medskip

The main novelty of this paper is threefold. First, we identify a noncommutative centered two-body fluctuation whose mixed moments enjoy cancellations analogous to those in the classical relative entropy method of Jabin and the third-named author \cite{jabin2016mean,jabin2018quantitative}. Second, this cancellation allows us to obtain a full $N$-body quantum relative entropy stability estimate for general mixed states, which in turn yields propagation of chaos in trace norm for fixed marginals. Third, we show that the same entropy mechanism extends to finite-dimensional open quantum systems governed by Lindblad dynamics, with general bounded two-body interactions.

\subsection{Organization of the Paper.}

The remainder of this paper is organized as follows. In Section \ref{sec: main}, we formulate our three main results: Theorems \ref{tm: relative entropy bound}, \ref{tm: uniform in h}, and \ref{thm: finite main}. We note that Theorem \ref{tm: uniform in h} is presented as a consequence of Theorem \ref{tm: relative entropy bound}, while the proof of Theorem \ref{thm: finite main} follows a similar strategy to that of Theorem \ref{tm: relative entropy bound}, with the necessary modifications to handle the dissipative part. 

For the sake of structural clarity, Section \ref{sec: relative entropy proof} provides an outline of the proof framework for Theorem \ref{tm: relative entropy bound} and identifies the essential intermediate steps. These intermediate results are rigorously established in Sections \ref{sec: evolution entropy} through \ref{sec: propagation}. Additionally, Section \ref{sec: QRE} recalls various fundamental properties of the quantum relative entropy and the preliminary lemmas required for our subsequent computations. Finally, Sections \ref{sec: uniformity} and \ref{sec: Lindblad} are dedicated to the proofs of Theorems \ref{tm: uniform in h} and \ref{thm: finite main}, respectively.

\subsection*{Acknowledgments} H. Liang thanks Jacky J. Chong for helpful discussions.  This work was partially supported by the National Key R\&D Program of China (Project No.~2024YFA1015500).  G. Guo is supported by the ANR project MATH-SPA. H. Liang and Z. Wang were partially supported by the NSFC (Grant Nos.~12595282 and 12171009). 

\bigskip
\section{Main Results}\label{sec: main}

We use a common ambient framework that includes both finite- and infinite-dimensional single-particle spaces. Let 
\begin{equation*}
    \Omega\in\{\Lambda,\mathbb{T}^d,\R^d\},
\end{equation*}
where $\Lambda$ is a finite lattice with periodic boundary conditions, equipped with the counting measure, while $\mathbb{T}^d$ and $\R^d$ are equipped with Lebesgue measure. We write 
\begin{equation*}
    \gH:= L^2(\Omega),
\end{equation*}
for the corresponding single-particle Hilbert space. In the case $\Omega=\Lambda$, this is canonically identified with the finite-dimensional space $\CC^{|\Lambda|}$, whereas for $\Omega=\mathbb{T}^d$ or $\R^d$ it is infinite-dimensional. For each $N\geq 1$, let $\gH_N:=\gH^{\otimes N}\simeq L^{2}(\Omega^N)$ be the $N$-particle Hilbert space.

Let $\cL(\gH)$ (resp. $\cL^{1}(\gH)$) be the space of bounded (resp. trace-class) operators on $\gH$. We denote by $\cD(\gH)$ the set of density operators on $\gH$, i.e. operators $\gamma$ satisfying 
\begin{equation*}
    \gamma=\gamma ^{*}\geq 0,\quad \tr_{\gH}(\gamma)=1.
\end{equation*}
Since we only consider indistinguishable particles, it is natural to consider only symmetric density operators on $\gH_N$. To be precise, for each permutation $\pi\in\gS_N$, let $U_{\pi}$ be the unitary operator on $\gH_N$ defined by
\begin{equation*}
    (U_{\pi}\Psi_N)(x_1,\cdots,x_N):=\Psi_N(x_{\pi^{-1}(1)},\cdots,x_{\pi^{-1}(N)}).
\end{equation*}
We say that an $N$-particle density operator $\Gamma$ on $\gH_N$ is symmetric if
\begin{equation*}
    U_{\pi}\Gamma U_{\pi}^{*}=\Gamma, \quad {\rm for \ all}\ \pi\in\gS_N.
\end{equation*}
We denote by $\cD_s(\gH_N)\subset \cD(\gH_N)$ the set of symmetric $N$-particle density operators on $\gH_N$. Let $\Gamma^N\in\cD_s(\gH_N)$ and $k\leq N$. We define the $k$-th marginal $\Gamma^{N:k}$ of $\Gamma^N$ by taking the partial trace over the remaining $N-k$ particles:
\begin{equation*}
    \Gamma^{N:k}:=\tr_{k+1,\ldots,N}(\Gamma^N).
\end{equation*}

\subsection{Quantum Mean-Field Limit: Fixed $\hbar$ and Large $N$.}
In this section, we consider the continuous model, and we choose $\Omega=\mathbb{T}^d$ or $\R^d$. Our quantum system is described by the following mean-field type Hamiltonian:
\begin{equation*}
    H_{N}:=\sum_{j=1}^{N}h_{j}+\frac{1}{N-1}\sum_{1\leq j<k\leq N}V(x_{j}-x_{k}),
\end{equation*}
where $h$ is a self-adjoint operator on $\gH$ bounded from below, and the interaction $V$ is a radially symmetric real-valued function satisfying $\norm{V}_{L^{\infty}(\Omega)}<\infty$. By the Friedrichs extension theorem, $H_{N}$ is a well-defined self-adjoint operator on $\gH_N$ with the same domain as $\sum_{j=1}^{N}h_j$. 

Given initial data $\Gamma_0^N\in \cD_s(\gH_N)$, the normalized state $\Gamma^N_t$ evolves according to the von Neumann equation:
\begin{equation*}
    \mathrm{i}\partial_t\Gamma_t^N=[H_N,\Gamma_t^N].
\end{equation*}
Since the Hamiltonian $H_{N}$ is symmetric and the evolution is unitary, one can check that $\Gamma_t^N\in\cD_s(\gH_N)$ for all $t\geq 0$. 

We will work under the following assumptions on the initial data $\gamma_0$:

\begin{assumption}\label{asump of gamma}
Assume that the initial data $\gamma_0$ satisfies:
\begin{enumerate}
    \item $\gamma_0$ is a faithful density operator, i.e. ${\rm Ker}\gamma_0=0$;
    \item $\gamma_0$ has finite entropy, i.e. $\gamma_0\log\gamma_0=(\log\gamma_0) \gamma_0$ is trace class;
    \item There exists a constant $C_1>0$ such that for any $f\in W^{1,\infty}(\Omega)$, $$\norm{[\log\gamma_0,f]}_{\rm op}\leq C_1\norm{|\nabla f|}_{L^{\infty}(\Omega)};$$
    \item There exists $C_2>0$ such that $$\norm{[\log\gamma_0,h]}_{\rm op}\leq C_2.$$
\end{enumerate}
\end{assumption}

\begin{remark}
    We present a typical example of $\gamma_0$ satisfying Assumption \ref{asump of gamma}. Consider a Toeplitz operator (see \eqref{eq:Toeplitz quantization}) $\gamma_0=\mathrm{Op}_{\hbar}^{T}(\mu)$ on $\mathbb{T}^d$ generated by a heavy-tailed Borel measure $\mu(\d q, \dd p)=\rho(p)\dd p\dd q$. For instance, choosing the momentum density as $\rho(p)=\mathfrak{C}(1+|p|^2)^{-M}$ with $M>\frac{d}{2}$ ensures that $\rho\in L^{1}(T^*(\mathbb{T}^d))$. One can check from the definition of Toeplitz quantization that $[\log\gamma_0,f]$ is an operator of order zero, satisfying the Lipschitz estimate in (3) of Assumption \ref{asump of gamma}.
\end{remark}

\begin{remark}
    By a standard commutator estimate, one can check that $\gamma_0=\frac{1}{Z}\exp(-(1-\Delta)^{s/2})$ with $s\in (0,1]$ also satisfies these assumptions.
\end{remark}

\bigskip

We will use the relative entropy method. Given two density operators $\Gamma$ and $\Gamma'$, we denote by $S(\Gamma,\Gamma')$ their quantum relative entropy. The definition and the properties of quantum relative entropy will be recalled in Section \ref{sec: QRE}. We now state our first main result.

\begin{theorem}\label{tm: relative entropy bound}
    Assume that $V\in W^{1,\infty}(\Omega)$ is a radially symmetric real-valued function and that the initial data $\gamma_0$ satisfies Assumption \ref{asump of gamma}. Let $\Gamma_t^N$ and $\gamma_t$ be the solutions to the global von Neumann equation and the Hartree equation, respectively:
\begin{equation*}
    \mathrm{i} \partial_t \Gamma_t^N = [H_{N}, \Gamma_t^N], \quad\quad
    \mathrm{i} \partial_t \gamma_t = [h + V \ast \rho_{\gamma_t}, \gamma_t],
\end{equation*}
where $\rho_{\gamma_t}(x) = \gamma_t(x, x)$ is the spatial density of $\gamma_t$. Then, for any fixed $T\geq 0$, there exists a constant $C$, depending only on $C_1, C_2, T$ and $\norm{V}_{W^{1,\infty}(\Omega)}$, such that 
\begin{equation}\label{eq:relative entropy bound}
    \sup_{t\in[0,T]}S(\Gamma_t^N,\gamma_t^{\otimes N})\leq C \left(S(\Gamma_0^N,\gamma_0^{\otimes N}) +\log 2\right). 
\end{equation}
The constant $C$ can be explicitly computed, and it is given by 
$$C= \exp\left(8C_0 C_1\norm{\nabla V}_{L^{\infty}}T+16C_0\left( C_1\norm{\nabla V}_{L^{\infty}}+C_2\right)\norm{V}_{L^{\infty}}T^2 \right),$$
where $C_0$ is the constant in Proposition~\ref{prop: combinatorics}.
\end{theorem}

\begin{remark}
    It is instructive to compare our entropic approach with the well-known projection method developed by Pickl \cite{pickl2011simple}. Pickl's method is based on the analysis of the time evolution of a counting operator for particles in excited states, utilizing a \textit{Gronwall argument} to derive trace-norm convergence via a functional inequality analogous to the Pinsker inequality. While the functional chosen by Pickl serves a mathematical purpose similar to that of the quantum relative entropy, the fundamental distinction lies in the class of admissible states. The projection method is essentially restricted to \textit{pure states}, whereas the quantum relative entropy framework is inherently suited for treating general \textit{mixed states}.
\end{remark}

\bigskip

By the block subadditivity of quantum relative entropy and Pinsker's inequality (see Lemma \ref{lem:block} and Lemma \ref{lem:pinsker}), we have 
\begin{equation*}
    S(\Gamma_t^{N:k},\gamma_t^{\otimes k})\leq\frac{k}{N}S(\Gamma^N_t,\gamma_t^{\otimes N})
\end{equation*}
and 
\begin{equation*}
    \norm{\Gamma_t^{N:k}-\gamma_t^{\otimes k}}_{1}\leq \sqrt{2S(\Gamma_t^{N:k},\gamma_t^{\otimes k})},
\end{equation*}
where $\norm{\cdot}_1$ is the trace-class norm. We thus arrive at the following result:

\begin{corollary}
    Assume the same hypotheses as in Theorem \ref{tm: relative entropy bound}. Let $\{\Gamma_0^N\}_{N\geq1} \subset \cD_s(\gH_N)$ be a sequence of initial data whose evolution is governed by the von Neumann equation. If there exists a constant $C_3$ such that
    \begin{equation*}
        \sup_{N}S(\Gamma_0^N,\gamma_0^{\otimes N})\leq C_3
    \end{equation*}
    then, for fixed $T\geq 0$, we have 
    \begin{equation*}
        \sup_{t\in[0,T]}\norm{\Gamma_t^{N:k}-\gamma_t^{\otimes k}}_1\leq \sqrt{\frac{2kC(C_3+\log2)}{N}}
    \end{equation*}
    for all $1\le k\le N$. Here $C$ is the same constant as in Theorem \ref{tm: relative entropy bound}, which depends only on $C_1, C_2, T$ and $\norm{V}_{W^{1,\infty}(\Omega)}$.
\end{corollary}

\bigskip

\subsection{Uniformity in $\hbar$.}

Let $\Omega=\mathbb{T}^d$ or $\R^d$. We consider the physically relevant case where the one-body operator is the quantum kinetic energy $-\hbar^2\Delta$, where $\h$ denotes the Planck constant, and $-\Delta$ is the Laplace operator on $\Omega$. 

We consider the evolution of a system of $N$ quantum particles interacting through an even two-body potential $\Phi$, described for any value of $\h>0$ by the $\h$-dependent von Neumann equation
\begin{equation}\label{eq:von-Neumann eqn-2}
    \mathrm{i}\hbar\partial_t\Gamma^N_t=[H_{N,\hbar},\Gamma^N_t],
\end{equation}
where the Hamiltonian $H_{N,\hbar}$ is given by
\begin{equation*}
    H_{N,\hbar}=\sum_{j=1}^{N}-\hbar^2\Delta_j+\frac{1}{N-1}\sum_{1\leq j<k\leq N}\Phi(x_j-x_k).
\end{equation*}
The corresponding mean-field equation is the $\h$-dependent Hartree equation
\begin{equation}\label{eq:Hartree eqn-2}
    \mathrm{i}\hbar\partial_t\gamma_t=[-\hbar^2\Delta+ \Phi\ast\rho_{\gamma_t},\gamma_t].
\end{equation}
Before stating our second main result, we recall the Husimi transform of a density operator and the corresponding Toeplitz quantization. 

We associate each point $z=(q,p)$ in the phase space $T^*\Omega\cong  \Omega\times \R^d$ with a vector $\ket{z}$ (coherent state) in the corresponding Hilbert space $L^2(\Omega)$ by
\begin{align*}
    \ket{z}(x):=
    \begin{cases}
    (\pi\hbar)^{-\frac{d}{4}}e^{-\frac{|x-q|^2}{2\hbar}}e^{\frac{\mathrm{i}p\cdot x}{\hbar}}, & \Omega=\R^d,\\
    \mathfrak{C}_z\sum_{k\in\Z^d}e^{-\frac{|x-k-q|^2}{2\hbar}}e^{\frac{\mathrm{i}p\cdot (x-k)}{\hbar}}, & \Omega=\bT^d.
    \end{cases}
\end{align*}
Here $\mathfrak{C}_z$ is a normalization constant such that $\norm{\ket{z}}_{L^2(\Omega)}=1$. We have the following resolution of identity:
\begin{equation}\label{eq:resolution of identity}
    \int_{T^*\Omega}\ket{z}\bra{z}\dd z=(2\pi\hbar)^d\1_{L^2(\Omega)}.
\end{equation}
Here and in the sequel, we use the Dirac bra-ket notation:
$$\bra{z}\Gamma\ket{z}:=\inprod{\ket{z}}{\Gamma\ket{z}}_{L^2(\Omega)},$$
and $\ket{z}\bra{z}$ is the orthogonal projector on $\ket{z}$. 

The Husimi transform of a density operator $\Gamma\in\cD(L^2(\Omega))$ is defined by
\begin{equation}\label{eq:Husimitrans}
    \mathscr{H}_{\hbar}[\Gamma](z):=\frac{1}{(2\pi\hbar)^d}\bra{z}\Gamma\ket{z},
\end{equation}
which associates to $\Gamma$ a probability density on the phase space $T^*\Omega$. Conversely, given a positive Borel measure $\mu$ on $T^*\Omega$, we can define the corresponding Toeplitz quantization by
\begin{equation}\label{eq:Toeplitz quantization}
    \mathrm{Op}_{\hbar}^T[\mu]:=\frac{1}{(2\pi\hbar)^d}\int_{T^*\Omega}\ket{z}\bra{z}\mu(\d z).
\end{equation}
We call an operator $\Gamma$ on $L^2(\Omega)$ a Toeplitz operator if there exists a positive Borel measure $\mu$ on $T^*\Omega$ and a positive constant $\eps_0$, such that $\Gamma=\mathrm{Op}_{\eps_0}^T[\mu]$.

One can check that $\tr(\mathrm{Op}_{\hbar}^T[\mu])=(2\pi\hbar)^{-d}\mu(T^*\Omega)$, which means the Toeplitz quantization maps Borel probability measures on the phase space to density operators on $L^2(\Omega)$ up to a constant $(2\pi\hbar)^d$. Elementary computations show that the Husimi transform and the Toeplitz quantization are dual to each other in the following sense:
\begin{equation*}
    \mathscr{H}_{\hbar}[\mathrm{Op}_{\hbar}^T[\mu]](z)=(2\pi\hbar)^{-d}\int_{T^*\Omega}e^{-\frac{|z-z'|^2}{2\hbar}}\mu(\d z').
\end{equation*}
Formally, the right-hand side converges to $\mu$ as $\hbar\rightarrow 0$. 

We will use Theorem \ref{tm: relative entropy bound}, together with the results in \cite{golse2016mean} and \cite{golse2018derivation}, to provide an $\hbar$-uniform result for the mean-field limit. We adopt the notations in \cite{golse2018derivation}. To characterize the convergence of $\Gamma_t$ to $\gamma_t^{\otimes N}$ uniformly in $\hbar$, we use the Wasserstein distance associated with the cost function $(z,z')\mapsto \min(1,|z-z'|)$ on the $j$-particle phase space $T^*(\Omega^j)\cong\Omega^j\times \R^{dj}$. 

To be precise, for Borel probability measures $\mu,\nu$ on $T^*(\Omega^j)$, the Wasserstein distance between $\mu$ and $\nu$ is defined by
\begin{equation*}
    \mathrm{dist}_1(\mu,\nu):=\inf_{\pi\in\Pi(\mu,\nu)}\int_{(T^*(\Omega^j))^2}\min(1,|z-z'|)\pi(\d z,\dd z'),
\end{equation*}
where $\Pi(\mu,\nu)$ is the set of all Borel probability measures on $T^*(\Omega^j)\times T^*(\Omega^j)$ with marginals $\mu$ and $\nu$.

\bigskip

Now we can state our second main result.
\begin{theorem}\label{tm: uniform in h}
    Assume that the interaction potential $\Phi$ is a real-valued, spherically symmetric function defined on $\Omega$, with $\Phi \in W^{1, \infty}(\Omega)$ and Lipschitz gradient $\nabla \Phi$.  
    
    Assume that $\gamma_0$ is a Toeplitz operator satisfying Assumption \ref{asump of gamma} with $h$ replaced by $-\Delta$, the Laplace operator on $\Omega$. Let $\gamma_t$ be the solution of the $\hbar$-dependent Hartree equation \eqref{eq:Hartree eqn-2} with initial data $\gamma_0$. For each $N\geq 1$, let $\Gamma^N_t$ be the solution of the $\hbar$-dependent von Neumann equation \eqref{eq:von-Neumann eqn-2} with initial data $\Gamma^N_0=\gamma_0^{\otimes N}$, and let $\Gamma_t^{N:k}$ be its $k$-th marginal for each $k=1,\dots,N$. Then, for any fixed $T\geq 0$ and $k\in\N$, we have
    \begin{equation*}
        \sup_{\hbar>0}\sup_{t\in[0,T]}\mathrm{dist}_1\left(\mathscr{H}_{\hbar}[\Gamma_t^{N:k}],\mathscr{H}_{\hbar}[\gamma_t^{\otimes k}]\right)^2\lesssim \frac{kdT\norm{\Phi}_{W^{1,\infty}}\left( e^{\Lambda T} +1\right) }{\left( \log N \right)^{\frac{1}{2}} }
    \end{equation*}
    for $N$ large enough. Here $\Lambda:=3+4\mathrm{Lip}(\nabla\Phi)^2$ is a constant depending only on $\Phi$.
\end{theorem}

It is useful to compare Theorem \ref{tm: uniform in h} with Theorem 3.1 of Golse--Paul--Pulvirenti \cite{golse2018derivation}, which builds on the quantum Monge--Kantorovich framework of Golse--Mouhot--Paul \cite{golse2016mean}. Under Toeplitz-type initial data and Lipschitz-force assumptions, Theorem 3.1 in \cite{golse2018derivation} gives an $\hbar$-uniform propagation-of-chaos estimate for the Husimi transforms in the bounded-cost distance $\mathrm{dist}_1$, with convergence rate of order $1/\sqrt{\log\log N}$ for $\mathrm{dist}_1$ itself, or equivalently order $1/\log\log N$ for $\mathrm{dist}_1^2$. Its proof interpolates an $\hbar$-dependent BBGKY estimate with the semiclassical quantum Wasserstein estimate from \cite{golse2016mean}. Our Theorem \ref{tm: uniform in h} has the same $\hbar$-uniform objective and is formulated in the same bounded-cost topology, but the fixed-$\hbar$ input is instead provided by the entropy estimate of Theorem \ref{tm: relative entropy bound}. Combining this entropy control with the quantum Wasserstein/Husimi estimate of \cite{golse2016mean} and optimizing in $\hbar$ yields the bound of order $(\log N)^{-1/2}$ for $\mathrm{dist}_1^2$ stated above. Thus, compared with Theorem 3.1 of \cite{golse2018derivation}, our argument replaces the BBGKY part of the interpolation by the relative entropy method and gives an improved logarithmic rate, under the assumptions of Theorem \ref{tm: uniform in h}.

The proof of Theorem \ref{tm: uniform in h} is given in Section \ref{sec: uniformity}.

\bigskip

\subsection{Open Lindbladian Systems.}

Let $\gH$ be a finite-dimensional complex Hilbert space. Let $h\in\cL(\gH)$ be a one-body self-adjoint operator, and let
$W=W^*\in\cL(\gH\otimes\gH)$ be a two-body self-adjoint interaction. We assume that $W$ is symmetric under exchange of the two tensor factors, namely
\[
    \mathsf{S}W\mathsf{S}=W,
\]
whenever $\mathsf{S}(u\otimes v)=v\otimes u$ denotes the flip operator on $\gH\otimes\gH$. We define the $N$-body Hamiltonian by

\begin{equation*}
    H_N:=\sum_{j=1}^N h_j+\frac{1}{N-1}\sum_{1\leq i<j\leq N}W_{ij}.
\end{equation*}
Here $W_{ij}$ denotes the operator $W$ acting on the $i$-th and $j$-th tensor factors.

For $L\in\cL(\gH)$, we define the one-body dissipator
\begin{equation*}
    \cL_{L}(A):= LAL^*-\frac{1}{2}\left( L^* LA+AL^* L \right). 
\end{equation*}
Given $\Gamma_0^N\in\cD_s(\gH_N)$, we consider the $N$-body Lindblad equation
\begin{equation}\label{eq:finite-Nbody-Lindblad}
    \partial_t\Gamma_t^N=-\ii [H_N,\Gamma_t^N]+\sum_{j=1}^N \cL_{L_j}(\Gamma_t^N),
\end{equation}
where $L_j$ is the operator obtained by letting $L$ act on the $j$-th variable. On the other hand, given $\gamma_0\in \cD(\gH)$, we consider the one-body Hartree--Lindblad equation
\begin{equation}\label{eq:finite-Hartree-Lindblad}
    \partial_t\gamma_t
    =
    -\ii[h+V^{\gamma_t},\gamma_t]+\cL_L(\gamma_t),
\end{equation}
where the mean-field potential is defined by
\begin{equation*}
    V^\gamma:=\tr_2\big((\1\otimes \gamma)W\big).
\end{equation*}
By the exchange symmetry of $W$, the same one-body operator is obtained by tracing over the first tensor factor:
\[
    V^\gamma=\tr_1\big((\gamma\otimes\1)W\big),
\]
after the natural identification of the two copies of $\gH$.

\bigskip

\begin{theorem}\label{thm: finite main}
Assume $LL^*=L^*L$ and there exists $m_0>0$ such that $\gamma_0\geq m_0 \1$. Let $\Gamma_t^N$ be a sufficiently regular solution of \eqref{eq:finite-Nbody-Lindblad} with initial data $\Gamma_0^N\in\cD_s(\gH_N)$, and let $\gamma_t$ be a sufficiently regular solution of \eqref{eq:finite-Hartree-Lindblad}. Define
\begin{equation}\label{eq:finite-def-CW-general}
    C_W^{(m_0)}
    :=
    \sup_{\substack{\gamma\in\cD(\gH)\\ \gamma\geq m_0\1}}
    \left\|
    -\ii\left[
    W-\tr_2((\1\otimes\gamma)W)\otimes\1,\,
    (\log\gamma)\otimes\1
    \right]
    \right\|_{\rm op}.
\end{equation}
Then $C_W^{(m_0)}<\infty$, and for any $t\geq0$ we have
\begin{equation*}
S(\Gamma_t^N,\gamma_t^{\otimes N})
    \leq
    e^{32 C_W^{(m_0)}t}
    \Big(S(\Gamma_0^N,\gamma_0^{\otimes N})+\log 2\Big).
\end{equation*}
Moreover, one has the explicit estimate
\begin{equation*}
    C_W^{(m_0)}\leq \frac{4\norm{W}_{\rm op}}{m_0}.
\end{equation*}
Consequently,
\begin{equation*}
S(\Gamma_t^N,\gamma_t^{\otimes N})
    \leq
    e^{\frac{128\norm{W}_{\rm op}t}{m_0}}
    \Big(S(\Gamma_0^N,\gamma_0^{\otimes N})+\log 2\Big).
\end{equation*}
\end{theorem}

\begin{remark}
The tensor-product assumption
\[
    W=\sum_{\alpha=1}^M F_\alpha\otimes G_\alpha
\]
is not needed for the entropy estimate itself. It is only a convenient sufficient condition for obtaining an explicit bound on the fluctuation operator. Indeed, if such a decomposition is chosen, then
\[
    C_W^{(m_0)}
    \leq
    \frac{4}{m_0}
    \sum_{\alpha=1}^M
    \norm{F_\alpha}_{\rm op}\norm{G_\alpha}_{\rm op}.
\]
Thus the previous form of the constant is recovered as a special case. In finite dimensions, every two-body operator admits a finite tensor-product expansion, but the entropy argument does not require the first tensor factors to belong to a fixed commutative algebra.
\end{remark}

\begin{remark}\label{rmk:finite-uniform-lower-bound-gammaT}
    Under the assumption that $\gamma_0\geq m_0 \1$, we have 
    \begin{equation*}
      \gamma_t\geq m_0\1,\qquad \forall t\geq 0.
    \end{equation*}
    Indeed, let
\[
\widetilde\gamma_t:=\gamma_t-m_0\1.
\]
Since $\1$ commutes with all operators and $\cL_L(\1)=0$, we see that $\widetilde\gamma_t$ satisfies the same evolution equation as $\gamma_t$: 
\[
\partial_t\widetilde\gamma_t
=
-\ii[h+V^{\gamma_t},\widetilde\gamma_t]+\cL_L(\widetilde\gamma_t).
\]
Because $\widetilde\gamma_0\geq 0$ and the evolution is positivity preserving, we conclude that $\widetilde\gamma_t\geq 0$, proving the claim.
\end{remark}

\begin{remark}\label{rmk:finite-propagation-of-chaos}
   As in the continuous case, we can also obtain convergence of the $k$-th marginal $\Gamma_t^{N:k}$ to $\gamma_t^{\otimes k}$ in trace norm, with an explicit rate depending on $N$ and $k$. We omit the details here. 
\end{remark}

The proof of Theorem \ref{thm: finite main} is given in Section \ref{sec: Lindblad}.

\bigskip
\section{Quantum Relative Entropy Method: Proof of Theorem \ref{tm: relative entropy bound}.}\label{sec: relative entropy proof}

This section gives the structure of the proof of Theorem \ref{tm: relative entropy bound}. We first derive the entropy production identity and rewrite the production term as an average of two-body fluctuation operators $X_{ij}$. We then prove a cancellation rule for mixed moments of these fluctuations, combine it with a combinatorial counting estimate, and establish the propagation of the commutator regularity needed to keep $X(t)$ bounded along the Hartree flow. These ingredients are finally inserted into the entropy inequality and closed by a Gronwall argument.

\subsection{Evolution of Quantum Relative Entropy.}
Define
\begin{equation*}
    \delta H_N(t)=\frac{1}{N-1}\sum_{1\leq j<k\leq N}V(x_j-x_k)- \sum_{j=1}^{N}(V\ast\rho_{\gamma_t})_j.
\end{equation*}
We have the following result.

\begin{proposition}\label{prop: evolution entropy}
    Let $\Gamma^N_t$, $\gamma_t$, and $V$ be as in Theorem \ref{tm: relative entropy bound}. We have 
    \begin{equation*}
        \frac{\dd}{\dd t}S(\Gamma_t^N,\gamma_t^{\otimes N})=\tr_{\gH_N}\left( \Gamma_t^N A_t^N \right) \,,
    \end{equation*}
    where 
    \begin{equation*}
        A_t^N=-\ii \left[\delta H_N(t),\log\gamma_t^{\otimes N} \right] .
    \end{equation*}
\end{proposition}

The proof of Proposition \ref{prop: evolution entropy} is given in Section \ref{sec: evolution entropy}. We then decompose $A_t^N$ into a sum of two-body fluctuations. For simplicity, we omit the subscript $t$ and use the shorthand $V^{\gamma}=V\ast\rho_{\gamma}$, $V_{ij}=V(x_i-x_j)$.

Using $\log\gamma^{\otimes N}=\sum_{k=1}^N \log\gamma_k$ and $\delta H = \frac{1}{N-1}\sum_{i<j}V_{ij}-\sum_i(V^\gamma)_i$,
we expand
\[
A=-\ii\comm{\delta H}{\log\gamma^{\otimes N}}
=
-\ii\comm{\frac{1}{N-1}\sum_{i<j}V_{ij}}{\sum_k (\log\gamma)_k}
+\ii\comm{\sum_i (V^\gamma)_i}{\sum_k (\log\gamma)_k}.
\]
Since $V_{ij}$ only acts on legs $i,j$, one has $[V_{ij},(\log\gamma)_k]=0$ unless $k\in\{i,j\}$, hence
\[
-\ii\comm{\frac{1}{N-1}\sum_{i<j}V_{ij}}{\sum_k (\log\gamma)_k}
=
-\frac{\ii}{N-1}\sum_{i<j}\Big(\comm{V_{ij}}{(\log\gamma)_i}+\comm{V_{ij}}{(\log\gamma)_j}\Big).
\]
Also, since $(V^\gamma)_i$ commutes with $(\log\gamma)_k$ unless $k=i$,
\[
\ii\comm{\sum_i (V^\gamma)_i}{\sum_k (\log\gamma)_k}
=
\sum_{i=1}^N \ii\comm{(V^\gamma)_i}{(\log\gamma)_i}.
\]
Therefore,
\[
A
=
-\frac{\ii}{N-1}\sum_{i<j}\Big(\comm{V_{ij}}{(\log\gamma)_i}+\comm{V_{ij}}{(\log\gamma)_j}\Big)
+\sum_{i=1}^N \ii\comm{(V^\gamma)_i}{(\log\gamma)_i}.
\]
Now use the identity $$\frac{1}{N-1}\sum_{i<j}\left(\ii\comm{(V^\gamma)_i}{(\log\gamma)_i}+\ii\comm{(V^\gamma)_j}{(\log\gamma)_j}\right)=\sum_{i=1}^N \ii\comm{(V^\gamma)_i}{(\log\gamma)_i}$$
and $V_{ij}=V_{ji}$. Grouping terms yields 
\begin{equation*}
    A=\frac{1}{N-1}\sum_{1\leq i\neq j\leq N}X_{ij}
\end{equation*}
where $X$ is a two-body fluctuation operator defined by 
\begin{equation}\label{eq: def of two-body fluctuation}
    X=-\ii \left[ V-(V\ast\rho_{\gamma})\otimes\1,(\log\gamma)\otimes\1 \right].
\end{equation}
Here we view the first $V$ in the bracket as a two-body multiplication operator defined by $V(x-y)$.

\begin{remark}
    The operator $X_{ij}$ should be compared with the centered two-particle fluctuation $\phi(x_i,x_j)$ appearing in the classical relative entropy method of Jabin and the third-named author \cite{jabin2016mean,jabin2018quantitative}. In the classical setting, subtracting the averaged interaction produces cancellations under the tensorized reference measure. Here $X_{ij}$ is the corresponding non-commutative object: the mean-field contribution is subtracted at the operator level, and the resulting centered quantity is then commuted with $\log\gamma$.
\end{remark}

\bigskip

\subsection{General Cancellation Rule for the Fluctuation Operator.}
Let $\gamma\in\cD(\gH)$ be faithful with finite entropy, and let $V$ be bounded. We also assume that the fluctuation operator $X$ defined as in \eqref{eq: def of two-body fluctuation} is bounded, so that \eqref{eq:cancellation} below is well-defined. 

Since the diagonal part of $V$ is subtracted, one can expect a certain cancellation rule for $X$. We have the following result.

\begin{proposition}[General cancellation rule]
Let $(I_m,J_m)$ be any pair of multi-indices, where $I_m=(i_1,\dots,i_m)$ and $J_m=(j_1,\dots,j_m)$ are two sequences of integers in $\{1,\dots,N\}$. We assume that $i_\nu\neq j_\nu$ for all $1\leq \nu\leq m$. Then we have
\begin{equation}\label{eq:cancellation}
    \tr_{\gH_N} \left( \gamma^{\otimes N}X_{i_{1}j_{1}}X_{i_{2}j_{2}}\cdots X_{i_{m}j_m} \right)=0 
\end{equation}
provided that one of the following statements is satisfied:
\begin{enumerate}
    \item There exists an index $i_{\nu}\notin \{i_1,\cdots,i_{\nu-1},i_{\nu+1},\cdots,i_m\}\cup \{j_1,\cdots,j_m\} $; 
    \item There exists an index $j_{\nu}\notin \{j_1,\cdots,j_{\nu-1},j_{\nu+1},\cdots,j_m\}\cup \{i_1,\cdots,i_m\} $.
\end{enumerate}
\label{prop:cancellation}
\end{proposition}

The proof of Proposition \ref{prop:cancellation} is given in Section \ref{sec: cancellation}. This result indicates that numerous terms vanish when computing the mixed moment of $X$ in the state $\gamma^{\otimes N}$. Moreover, we have the following combinatorial result.

\begin{proposition}\label{prop: combinatorics}
Given positive integers $m, N$, let $\mathcal{I}_{m,N}$ denote the set of pairs of multi-indices $(I_m, J_m)$ where $I_m = (i_1, \dots, i_m)$ and $J_m = (j_1, \dots, j_m)$ with $i_\nu, j_\nu \in \{1, \dots, N\}$ and such that the following conditions hold: 
\begin{enumerate}
    \item $i_\nu \neq j_\nu$ for all $\nu \in \{1, \dots, m\}$.
    \item For every $\nu$, either there exists $\nu'$ such that $i_\nu = j_{\nu'}$, or there exists $\nu'' \neq \nu$ such that $i_\nu = i_{\nu''}$.
    \item For every $\nu$, either there exists $\nu'$ such that $j_\nu = i_{\nu'}$, or there exists $\nu'' \neq \nu$ such that $j_\nu = j_{\nu''}$.
\end{enumerate}
Then there exists a universal constant $C_0$ such that
\begin{equation*}
    |\cI_{m,N}| \leq (C_0N)^m m!.
\end{equation*}
\end{proposition}

The proof of Proposition \ref{prop: combinatorics} is given in Section \ref{sec: cancellation}.

\bigskip

\subsection{Propagation of Regularity.}
Note that we assume a priori in Proposition \ref{prop:cancellation} that $X$ is bounded. Since $X$ depends on $t$ via $\gamma_t$, where $\gamma_t$ satisfies the Hartree equation, we need to show that $X$ remains a bounded operator under the Hartree flow. Using the assumptions on the initial data $\gamma_0$ in Assumption \ref{asump of gamma}, one can check that $X(0)$ is a bounded operator.  

\begin{proposition}[Propagation of regularity]\label{prop: propagation of regularity}
    Assume the hypotheses of Theorem \ref{tm: relative entropy bound}. There exists an increasing function $C(t)$ such that
    \begin{equation*}
        \norm{X(t)}_{\mathrm{op}}\leq C(t).
    \end{equation*}
More precisely, we can choose $C(t)=2C_1\norm{\nabla V}_{L^{\infty}}+4 \left( C_1\norm{\nabla V}_{L^{\infty}}+C_2 \right)\norm{V}_{L^{\infty}}t$, where $C_1$ and $C_2$ are the constants in Assumption \ref{asump of gamma}. 
\end{proposition}

The proof of Proposition \ref{prop: propagation of regularity} is given in Section \ref{sec: propagation}.

\bigskip

\subsection{Proof of Theorem \ref{tm: relative entropy bound}.}

With Propositions \ref{prop: evolution entropy}--\ref{prop: propagation of regularity} in place, we are ready to prove the main result, Theorem \ref{tm: relative entropy bound}. The argument below combines the entropy production identity, the cancellation and counting estimates for the fluctuation operator, and the propagated operator-norm bound for $X(t)$.

\begin{proof}[Proof of Theorem \ref{tm: relative entropy bound}]
    By Proposition \ref{prop: evolution entropy}, we have 
    \begin{equation*}
        \frac{\dd}{\dd t}S(\Gamma_t^N,\gamma_t^{\otimes N})=\tr_{\gH_N}\left( \Gamma_t^N A_t^N \right) \,,
    \end{equation*}
    \begin{equation*}
        A_t^N=\frac{1}{N-1}\sum_{1\leq i\neq j\leq N}X_{ij}(t),
    \end{equation*}
    where $X(t)$ is defined in \eqref{eq: def of two-body fluctuation}. By Proposition \ref{prop: propagation of regularity}, $X(t)$ is a bounded operator for any $t\ge 0$, whose operator norm is bounded by a continuous function $C(t)$. This allows us to use the following entropy inequality.

    \begin{lemma}[Entropy inequality]\label{lem:entropy ineq}
    Let $\rho\in \cD(\gH)$ be a density operator on a Hilbert space, $\sigma\in\cD(\gH)$ be faithful, and $A$ be a bounded self-adjoint operator on $\gH$. Then, for any $\lambda>0$, we have
    \begin{equation*}
        \tr_{\gH}(\rho A)\leq\frac{1}{\lambda}S(\rho,\sigma)+\frac{1}{\lambda}\log\tr_{\gH}(\sigma e^{\lambda A}).
    \end{equation*}
    \end{lemma}
    Using Lemma \ref{lem:entropy ineq}, we have 
    \begin{align*}
        \frac{\dd}{\dd t}S(\Gamma_t^N,\gamma_t^{\otimes N})&=\tr_{\gH_N}\left( \Gamma_t^N A_t^N \right)\\
        &\leq \frac{1}{\lambda}S(\Gamma_t^N,\gamma_t^{\otimes N})+\frac{1}{\lambda}\log\tr_{\gH_N}\left( \gamma_t^{\otimes N}e^{\lambda A_t^N} \right).  
    \end{align*}
    Here $\lambda$ is a $t$-dependent parameter to be determined later. Since $A_t^N$ is bounded, we use Taylor expansion to obtain
    \begin{equation*}
        \tr_{\gH_N}\left( \gamma_t^{\otimes N}e^{\lambda A_t^N} \right)=\sum_{m=0}^{\infty}\frac{\lambda^m}{m!}\tr_{\gH_N}\left( \gamma_t ^{\otimes N}(A_t^N)^{m}\right).  
    \end{equation*}
    We omit the subscript $t$ for simplicity. For a fixed $m$, each term in $\tr(\gamma^{\otimes N}A^m)$ has the following form:
    \begin{equation}\label{eq:expansion of m-th mixed moment}
        \tr \left( \gamma^{\otimes N}A^m \right)=\frac{1}{(N-1)^m}\sum_{(I_m,J_m)}\tr_{\gH_N} \left( \gamma^{\otimes N}X_{i_{1}j_{1}}X_{i_{2}j_{2}}\cdots X_{i_{m}j_m} \right) ,
    \end{equation}
    where the multi-indices $(I_m,J_m)$ run over $\{1,\cdots,N\}^{2m}$ such that $i_{\nu}\neq j_{\nu}$. We know from the general cancellation rule in Proposition \ref{prop:cancellation} that numerous terms in \eqref{eq:expansion of m-th mixed moment} vanish. This, combined with the combinatorial result in Proposition \ref{prop: combinatorics}, yields
    \begin{equation*}
        \left|\tr(\gamma^{\otimes N}A^m)\right\vert\leq \frac{1}{(N-1)^m}|\cI_{m,N}|\norm{X}_{\rm op}^m\leq \left( 2C_0 C(t) \right)^m m!. 
    \end{equation*}
    Plugging this into the Taylor expansion, we have (for $\lambda$ small enough)
    \begin{equation*}
        \tr\left(\gamma^{\otimes N}e^{\lambda A}\right)\leq\sum_{m=0}^{\infty}\left( 2C_0C(t)\lambda \right)^m =\frac{1}{1-2C_0C(t)\lambda}.
    \end{equation*}
    Choosing $\lambda=(4C_0C(t))^{-1}$, we arrive at 
    \begin{equation*}
        \frac{\dd}{\dd t}S(\Gamma_t^N,\gamma_t^{\otimes N})\leq 4C_0C(t)S(\Gamma_t^N,\gamma_t^{\otimes N})+4(\log2)C_0C(t),
    \end{equation*}
    and \eqref{eq:relative entropy bound} follows from a Gronwall argument. Using $\sup_{[0,T]}C(t)=C(T)$ and Gronwall's inequality, we have for all $t\in[0,T]$
    \begin{equation}\label{eq:gronwall}
        S(\Gamma_t^N,\gamma_t^{\otimes N})\leq e^{4C_0C(T)t}S(\Gamma_0^N,\gamma_0^{\otimes N})+(\log2)e^{4C_0C(T)t}.
    \end{equation}
    This completes the proof of Theorem \ref{tm: relative entropy bound}.

\end{proof}

It remains to prove Lemma \ref{lem:entropy ineq}, which is equivalent to the Gibbs variational principle. Here we give a direct proof, using the nonnegativity of quantum relative entropy. 

\begin{proof}[Proof of Lemma \ref{lem:entropy ineq}]
\textbf{Step 1.}
Define the density operator
\[
\omega := \frac{e^{\log\sigma+\lambda A}}{\tr\big(e^{\log\sigma+\lambda A}\big)}\equiv \frac{e^{\log\sigma+\lambda A}}{Z}.
\]
Since $\sigma$ is faithful and $A$ is bounded self-adjoint, $\omega$ is well-defined. 
By nonnegativity of quantum relative entropy,
\[
0\le S(\rho,\omega)=\tr\big(\rho(\log\rho-\log\omega)\big).
\]
Using $\log\omega=(\log\sigma+\lambda A)-\log Z$, we obtain
\[
0\le \tr(\rho(\log\rho-\log\sigma-\lambda A+\log Z))
= S(\rho,\sigma) - \lambda \tr(\rho A) + \log Z.
\]
Rearranging gives
\begin{equation}\label{eq:entropy-ineq-Z}
\tr(\rho A)\le \frac{1}{\lambda}S(\rho,\sigma)+\frac{1}{\lambda}\log Z.
\end{equation}
\textbf{Step 2.}
By the Golden--Thompson inequality,
\[
\tr(e^{X+Y})\le \tr(e^{X}e^{Y}),
\]
with $X=\log\sigma$ and $Y=\lambda A$, one deduces
\[
Z=\tr\big(e^{\log\sigma+\lambda A}\big)\le \tr\big(e^{\log\sigma}e^{\lambda A}\big)=\tr(\sigma e^{\lambda A}).
\]
Hence, $\log Z\le \log\tr(\sigma e^{\lambda A})$. Plugging this into \eqref{eq:entropy-ineq-Z} yields the claimed inequality.
\end{proof}

\bigskip
\section{Quantum Relative Entropy Preliminaries.}\label{sec: QRE}

Throughout this section, all integrals are understood as Bochner integrals in the strong operator topology. The results in this section hold for any Hilbert space $\gH$.

For two density operators $\Gamma,\Gamma'$ on $\gH$, the quantum relative entropy is defined as
\begin{equation*}
    S(\Gamma,\Gamma')=\left\{\begin{array}{ll}
    \tr_{\gH}\left(\Gamma(\log \Gamma -\log \Gamma')\right),&\text{if }\Ker(\Gamma')\subset \Ker(\Gamma),\\
    +\infty,&\text{otherwise}.
    \end{array}\right.
\end{equation*}
The quantum relative entropy is always non-negative and equals zero if and only if $\Gamma=\Gamma'$. Moreover, it is monotone under the action of completely positive trace-preserving maps.

We first recall some useful properties of the logarithm.

\begin{lemma}[Fr\'echet derivative of $\log$]\label{lem:FrechetLog}
Let $X\in\cD(\gH)$ be a faithful density operator. Then, for any $B\in\cL(\gH)$, the Fr\'echet derivative $\mathrm D(\log)_X[B]$ is a densely-defined operator on $\gH$, which can be represented as
\[
\mathrm D(\log)_X[B]
=
\int_0^\infty (X+s\1)^{-1}\,B\,(X+s\1)^{-1}\,\dd s.
\]
In particular, if $t\mapsto X_t$ is differentiable,
\[
\partial_t(\log X_t) = \mathrm D(\log)_{X_t}[\partial_t X_t].
\]
\end{lemma}

\begin{proof}
For a scalar $x>0$, one has the identity
\begin{equation*}
\log x = \int_0^\infty \left(\frac{1}{1+s}-\frac{1}{x+s}\right)\,\dd s.
\end{equation*}
For every density operator $X>0$, the spectral theorem gives
\begin{equation}\label{eq:matrix-log-rep}
\log X = \int_0^\infty\left(\frac{1}{1+s}\1-(X+s\1)^{-1}\right)\dd s.
\end{equation}
Fix $s> 0$ and consider $F_s(X):=(X+s\1)^{-1}$.
For small $t$ (so that $X+tB+s\1$ remains invertible),
the resolvent identity gives
\[
(X+tB+s\1)^{-1}-(X+s\1)^{-1}
= -(X+s\1)^{-1}\,tB\,(X+tB+s\1)^{-1}.
\]
Dividing by $t$ and letting $t\to 0$ yields the Fr\'echet derivative
\begin{equation}\label{eq:Frechet-resolvent}
\mathrm D F_s(X)[B]
= - (X+s\1)^{-1}\,B\,(X+s\1)^{-1}.
\end{equation}
Let $G(X)$ denote the right-hand side of \eqref{eq:matrix-log-rep}.
Using \eqref{eq:Frechet-resolvent} and linearity,
\[
\mathrm D G(X)[B]
= \int_0^\infty (X+s\1)^{-1}\,B\,(X+s\1)^{-1}\,\dd s.
\]

If $t\mapsto X_t$ is differentiable, then $t\mapsto \log X_t$ is differentiable and
\[
\partial_t(\log X_t)=\mathrm D(\log)_{X_t}[\partial_t X_t],
\]
by the chain rule for Fr\'echet differentiable maps.
\end{proof}

\begin{lemma}[Commutator--log identity]\label{lem:comm-log}
Let $X$ be a faithful density operator and $B=B^\ast$ be bounded. Then
\[
\ii\comm{B}{\log X}
=
\int_0^\infty (X+s\1)^{-1}\,\big(\ii\comm{B}{X}\big)\,(X+s\1)^{-1}\,\dd s.
\]
\end{lemma}

\begin{proof}
\textbf{Step 1.}
Define the unitary family $U_\theta:=e^{\ii\theta B}$ and
\[
X_\theta := U_\theta X U_\theta^\ast.
\]
Then $X_\theta>0$ for all $\theta$, and differentiating gives
\[
\left.\frac{\dd}{\dd\theta}X_\theta\right|_{\theta=0}
=\ii B X - \ii X B
=\ii\comm{B}{X}.
\]
\textbf{Step 2.}
Since functional calculus is equivariant under unitary conjugation,
\[
\log X_\theta = U_\theta (\log X) U_\theta^\ast.
\]
Therefore,
\[
\left.\frac{\dd}{\dd\theta}\log X_\theta\right|_{\theta=0}
=\ii\comm{B}{\log X}.
\]
\textbf{Step 3.}
By Lemma~\ref{lem:FrechetLog} and the chain rule,
\[
\left.\frac{\dd}{\dd\theta}\log X_\theta\right|_{\theta=0}
=
\mathrm D(\log)_X\!\left[\left.\frac{\dd}{\dd\theta}X_\theta\right|_{\theta=0}\right]
=
\mathrm D(\log)_X\big[\ii\comm{B}{X}\big].
\]
Using the explicit Fr\'echet derivative formula from Lemma~\ref{lem:FrechetLog} yields
\[
\ii\comm{B}{\log X}
=
\int_0^\infty (X+s\1)^{-1}\,\big(\ii\comm{B}{X}\big)\,(X+s\1)^{-1}\,\dd s,
\]
as desired.
\end{proof}

\begin{lemma}[Derivative formula]\label{lem:RE-derivative}
Let $\Gamma_t,\sigma_t$ be differentiable families of density operators which are faithful with finite entropy. Then
\[
\frac{\dd}{\dd t}S(\Gamma_t,\sigma_t)
=
\tr\Big((\partial_t\Gamma_t)(\log\Gamma_t-\log\sigma_t)\Big)
-\tr\Big(\Gamma_t\,\partial_t\log\sigma_t\Big).
\]
\end{lemma}

\begin{proof}
Recall
\[
S(\Gamma_t,\sigma_t)=\tr(\Gamma_t\log\Gamma_t)-\tr(\Gamma_t\log\sigma_t).
\]
We treat the two terms separately.

\vspace{2mm}

\noindent \textbf{Step 1.}
Using the product rule,
\[
\frac{\dd}{\dd t}\tr(\Gamma_t\log\Gamma_t)
=\tr\big((\partial_t\Gamma_t)\log\Gamma_t\big)+\tr\big(\Gamma_t\,\partial_t\log\Gamma_t\big).
\]
We show that $\tr(\Gamma_t\,\partial_t\log\Gamma_t)=\tr(\partial_t\Gamma_t)$, which vanishes since
$\tr(\Gamma_t)=1$ for all $t$. By Lemma~\ref{lem:FrechetLog},
\[
\partial_t\log\Gamma_t = \mathrm D(\log)_{\Gamma_t}[\partial_t\Gamma_t]
= \int_0^\infty (\Gamma_t+s\1)^{-1}\,(\partial_t\Gamma_t)\,(\Gamma_t+s\1)^{-1}\,\dd s.
\]
Hence, using cyclicity of the trace,
\begin{align*}
\tr(\Gamma_t\,\partial_t\log\Gamma_t)
&=\int_0^\infty \tr\!\Big(\Gamma_t(\Gamma_t+s\1)^{-1}\,(\partial_t\Gamma_t)\,(\Gamma_t+s\1)^{-1}\Big)\,\dd s\\
&=\int_0^\infty \tr\!\Big((\Gamma_t+s\1)^{-1}\Gamma_t(\Gamma_t+s\1)^{-1}\,(\partial_t\Gamma_t)\Big)\,\dd s.
\end{align*}
Note that $(\Gamma_t+s\1)^{-1}\Gamma_t(\Gamma_t+s\1)^{-1}$ is a function of $\Gamma_t$ and thus diagonalizable in the same basis as $\Gamma_t$. It follows that 
\[
\int_0^\infty (\Gamma_t+s\1)^{-1}\Gamma_t(\Gamma_t+s\1)^{-1}\,\dd s = \1
\]
because for each scalar $x>0$,
\[
\int_0^\infty \frac{x}{(x+s)^2}\,\dd s
=1.
\]
Therefore,
\[
\tr(\Gamma_t\,\partial_t\log\Gamma_t)=\tr(\partial_t\Gamma_t).
\]
Since $\tr(\Gamma_t)=1$, we have $\tr(\partial_t\Gamma_t)=0$, and consequently
\[
\frac{\dd}{\dd t}\tr(\Gamma_t\log\Gamma_t)=\tr\big((\partial_t\Gamma_t)\log\Gamma_t\big).
\]
\textbf{Step 2.} Using the product rule again, one has
\[
\frac{\dd}{\dd t}\big(-\tr(\Gamma_t\log\sigma_t)\big)
=-\tr\big((\partial_t\Gamma_t)\log\sigma_t\big)-\tr\big(\Gamma_t\,\partial_t\log\sigma_t\big)
\]
and thus the desired identity.
\end{proof}

The following two properties of quantum relative entropy are essential for this paper.
\begin{lemma}[Quantum Pinsker]\label{lem:pinsker}
For density operators $\eta,\zeta$,
\[
\norm{\eta-\zeta}_1^2 \le 2 S(\eta,\zeta).
\]
\end{lemma}
This is the quantum analogue of Pinsker's inequality. See \cite{carlen2014remainder} and \cite{hayashi2006quantum}.

\begin{lemma}[Block subadditivity bound]\label{lem:block}
Assume $\Gamma\in\cD_s(\gH_N)$ and $\gamma\in\cD(\gH)$.
Then for each fixed $k\ge 1$,
\[
S(\Gamma^{N:k},\gamma^{\otimes k})
\le
\frac{k}{N}\,S(\Gamma,\gamma^{\otimes N}).
\]
\end{lemma}

\begin{proof}
    Since $$\log(\gamma^{\otimes k})=\sum_{j=1}^k (\log\gamma)_j,$$ and $\Gamma\in\cD_{s}(\gH_N)$ is permutation-invariant, we have 
    \begin{equation*}
        \tr \left( \Gamma^{N:k}(\log\gamma^{\otimes k}) \right) =k\tr \left( \Gamma^{N:1}\log\gamma \right). 
    \end{equation*}
    Therefore 
    \begin{equation*}
        \frac{1}{k}S(\Gamma^{N:k},\gamma^{\otimes k})
        =-\frac{1}{k}S \left( \Gamma^{N:k} \right)  -\tr \left( \Gamma^{N:1}\log\gamma \right)\,,
    \end{equation*}
    where $S(A)=-\tr(A\log A)$ is the von Neumann entropy of $A$. It suffices to show that 
    \begin{equation*}
        S(\Gamma)\leq\frac{N}{k}S \left( \Gamma^{N:k} \right). 
    \end{equation*}
    We will use the strong subadditivity of von Neumann entropy. Let $[N]$ be the set $\{1,2,\ldots,N\}$; we denote by $[N]_k$ the family of subsets of $[N]$ such that each subset has $k$ elements. For each $I\in [N]_k$, we denote by $\Gamma^I$ the reduced density operator of $\Gamma$ on the legs in $I$. Each index $i\in[N]$ belongs to exactly $\binom{N-1}{k-1}$ subsets in $[N]_k$. Using the strong subadditivity of von Neumann entropy, we have
    \begin{equation*}
        S(\Gamma)\leq\frac{1}{\binom{N-1}{k-1}}\sum_{I\in [N]_k}S(\Gamma^I).
    \end{equation*}
    For the strong subadditivity of von Neumann entropy and the above inequality, we refer to Theorem 3.1 in \cite{carlen2008brascamp} for example. Since $\Gamma$ is permutation-invariant, we have $S(\Gamma^I)=S(\Gamma^{N:k})$ for all $I\in [N]_k$. Hence,
    \begin{equation*}
        S(\Gamma)\leq\frac{1}{\binom{N-1}{k-1}}\sum_{I\in [N]_k}S(\Gamma^I)=\frac{\binom{N}{k}}{\binom{N-1}{k-1}}S(\Gamma^{N:k})=\frac{N}{k}S(\Gamma^{N:k}),
    \end{equation*}
    which completes the proof.
\end{proof}

\bigskip
\section{Evolution of Quantum Relative Entropy: Proof of Proposition \ref{prop: evolution entropy}.}\label{sec: evolution entropy}

The purpose of this section is to compute the time derivative of $S(\Gamma_t^N,\gamma_t^{\otimes N})$. The calculation shows that the entropy production is governed exactly by the Hamiltonian defect $\delta H_N(t)$ through its commutator with $\log\gamma_t^{\otimes N}$.

\begin{proof}
\medskip
\noindent\textbf{Step 1.}
We use the shorthand $\gamma^{\otimes N}=\sigma, V^{\gamma}=V\ast\rho_{\gamma}$ and omit the subscript $t$. By Lemma~\ref{lem:RE-derivative}, 
\begin{equation}\label{eq:thm-master-step1}
\frac{\dd}{\dd t}S(\Gamma,\sigma)
=
\tr\Big((\partial_t\Gamma)(\log\Gamma-\log\sigma)\Big)
-\tr\Big(\Gamma\,\partial_t\log\sigma\Big).
\end{equation}

\medskip
\noindent\textbf{Step 2.}
From the von Neumann equation,
\[
\partial_t\Gamma = -\ii\comm{H_N}{\Gamma}.
\]
Insert into \eqref{eq:thm-master-step1}:
\begin{align}
\frac{\dd}{\dd t}S(\Gamma,\sigma)
&=
\tr\Big(-\ii\comm{H_N}{\Gamma}(\log\Gamma-\log\sigma)\Big)
-\tr\Big(\Gamma\,\partial_t\log\sigma\Big). \label{eq:thm-master-step2}
\end{align}
We claim that
\begin{equation*}
\tr\Big(-\ii\comm{H_N}{\Gamma}\,\log\Gamma\Big)=0.
\end{equation*}
Indeed, using cyclicity of the trace,
\begin{align*}
\tr\Big(-\ii\comm{H_N}{\Gamma}\,\log\Gamma\Big)
&= -\ii\tr(H_N\Gamma\log\Gamma-\Gamma H_N\log\Gamma)
= -\ii\tr(H_N\Gamma\log\Gamma-H_N\log\Gamma\,\Gamma)\\
&= -\ii\tr\big(H_N(\Gamma\log\Gamma-\log\Gamma\,\Gamma)\big)
= -\ii\tr\big(H_N\comm{\Gamma}{\log\Gamma}\big)=0,
\end{align*}
since $\Gamma$ commutes with $\log\Gamma$.
Therefore,
\begin{align}
\tr\Big(-\ii\comm{H_N}{\Gamma}(\log\Gamma-\log\sigma)\Big)
= \tr\Big(\ii\comm{H_N}{\Gamma}\,\log\sigma\Big) \notag
= -\tr\Big(\Gamma\,\ii\comm{H_N}{\log\sigma}\Big),
\end{align}
where the last equality is again cyclicity. Then, \eqref{eq:thm-master-step2} becomes
\begin{equation}\label{eq:thm-master-step2b}
\frac{\dd}{\dd t}S(\Gamma,\sigma)
=
-\tr\Big(\Gamma\,\ii\comm{H_N}{\log\sigma}\Big)
-\tr\Big(\Gamma\,\partial_t\log\sigma\Big).
\end{equation}

\medskip
\noindent\textbf{Step 3.}
Since $\sigma=\gamma^{\otimes N}$, we have the exact identity
\begin{equation}\label{eq:log-tensor-sum}
\log\sigma = \sum_{j=1}^N (\log\gamma)_j.
\end{equation}
Hence,
\begin{equation}\label{eq:dlogsigma-sum}
\partial_t\log\sigma = \sum_{j=1}^N (\partial_t\log\gamma)_j.
\end{equation}
We now compute $\partial_t\log\gamma$. By Lemma~\ref{lem:FrechetLog} and the chain rule,
\[
\partial_t\log\gamma
=\mathrm D(\log)_\gamma[\partial_t\gamma]
=
\mathrm D(\log)_\gamma\!\Big[-\ii\comm{h+V^\gamma}{\gamma}\Big].
\]
Split into Hamiltonian and dissipative contributions:
\begin{equation}\label{eq:dloggamma-split}
\partial_t\log\gamma
=
\mathrm D(\log)_\gamma\!\Big[-\ii\comm{h+V^\gamma}{\gamma}\Big]
.
\end{equation}
For the Hamiltonian part we use Lemma~\ref{lem:comm-log} in the equivalent form
\[
\mathrm D(\log)_X[\ii\comm{B}{X}] = \ii\comm{B}{\log X}.
\]
Indeed, Lemma~\ref{lem:comm-log} states exactly that $\ii\comm{B}{\log X}$ equals the integral formula for
$\mathrm D(\log)_X[\ii\comm{B}{X}]$ given by Lemma~\ref{lem:FrechetLog}.
Applying this with $X=\gamma$ and $B=h+V^\gamma$ gives
\[
\mathrm D(\log)_\gamma\!\Big[-\ii\comm{h+V^\gamma}{\gamma}\Big]
=
-\ii\comm{h+V^\gamma}{\log\gamma}.
\]
Therefore, from \eqref{eq:dloggamma-split},
\begin{equation}\label{eq:dloggamma-final}
\partial_t\log\gamma
=
-\ii\comm{h+V^\gamma}{\log\gamma}.
\end{equation}
Plugging \eqref{eq:dloggamma-final} into \eqref{eq:dlogsigma-sum} yields
\begin{align}
\partial_t\log\sigma
&=
-\ii\sum_{j=1}^N \comm{h_j+(V^\gamma)_j}{(\log\gamma)_j}\notag\\
&= -\ii\comm{H^{\mathrm{mf}}_{N}(t)}{\log\sigma},
\label{eq:dlogsigma-final}
\end{align}
where we used \eqref{eq:log-tensor-sum} and the definition
$H^{\mathrm{mf}}_{N}(t)=\sum_{j=1}^N (h_j+(V^\gamma)_j)$.

\medskip
\noindent\textbf{Step 4.}
Insert \eqref{eq:dlogsigma-final} into \eqref{eq:thm-master-step2b}:
\begin{align}
\frac{\dd}{\dd t}S(\Gamma,\sigma)
&=
-\tr\Big(\Gamma\,\ii\comm{H_N}{\log\sigma}\Big)
-\tr\Big(\Gamma\Big[-\ii\comm{H^{\mathrm{mf}}_{N}(t)}{\log\sigma}\Big]\Big)\notag\\
&=
-\tr\Big(\Gamma\,\ii\comm{H_N}{\log\sigma}\Big)
+\tr\Big(\Gamma\,\ii\comm{H^{\mathrm{mf}}_{N}(t)}{\log\sigma}\Big).\label{eq:thm-master-step4}
\end{align}
Using linearity of the commutator,
\[
-\ii\comm{H_N}{\log\sigma}+\ii\comm{H^{\mathrm{mf}}_{N}(t)}{\log\sigma}
=
-\ii\comm{H_N-H^{\mathrm{mf}}_{N}(t)}{\log\sigma}
=
-\ii\comm{\delta H_N(t)}{\log\sigma},
\]
so \eqref{eq:thm-master-step4} becomes
\begin{equation*}
\frac{\dd}{\dd t}S(\Gamma,\sigma)
=
-\tr\Big(\Gamma\,\ii\comm{\delta H_N(t)}{\log\sigma}\Big),
\end{equation*}
which is the claimed identity. 
\end{proof}

\bigskip
\section{Cancellation Rule: Proof of Propositions \ref{prop:cancellation} and \ref{prop: combinatorics}.}\label{sec: cancellation}

This section proves the two structural estimates behind the partition-function bound used in the entropy argument. First, we show that any mixed moment containing an index that appears only once vanishes because of the centering built into $X_{ij}$. Second, we count the remaining index configurations and prove that their number is small enough to conclude the uniform-in-$N$ bound for partition functions as in Lemma \ref{lem:entropy ineq}.

\begin{proof}[Proof of Proposition \ref{prop:cancellation}]
    Without loss of generality, we assume that condition (1) is satisfied. We can write
\begin{equation*}
    \tr \left( \gamma^{\otimes N}X_{i_{1}j_{1}}X_{i_{2}j_{2}}\cdots X_{i_{m}j_m} \right)=\tr \left( \gamma^{\otimes N}ZX_{i_{\nu}j_{\nu}}Y \right),
\end{equation*}
where $Y,Z$ are products of operators which do not act on the $i_{\nu}$-th variable. Taking the partial trace with respect to the $i_{\nu}$-th variable, we have
\begin{equation*}
    \tr \left( \gamma^{\otimes N}ZX_{i_{\nu}j_{\nu}}Y \right)=\tr \left( \gamma^{\otimes (N-1)}Z \tr_{i_{\nu}}\left( \gamma X_{i_{\nu}j_{\nu}} \right) Y \right).
\end{equation*}
Note that
\begin{align*}
    \tr_{i_{\nu}}\left( \gamma X_{i_{\nu}j_{\nu}} \right)=
    \begin{cases}
    -\mathrm{i}\tr_1\left( \gamma\otimes\1 \left[ V-V^{\gamma}\otimes\1,(\log\gamma)\otimes\1 \right] \right),\ {\rm if}\ {i_{\nu}<j_{\nu}};\\
    -\mathrm{i}\tr_1\left( \gamma\otimes\1 \left[ V-\1\otimes V^{\gamma},\1\otimes\log\gamma \right] \right),\ {\rm if}\ {i_{\nu}>j_{\nu}}.
    \end{cases}
\end{align*}
Under the assumption that $\gamma\log\gamma=(\log\gamma)\gamma $ is trace-class and $V$ is bounded, we have
\begin{align*}
-\mathrm{i}\tr_1\left( \gamma\otimes\1 \left[ V-V^{\gamma}\otimes\1,\log\gamma\otimes\1 \right] \right)=\mathrm{i}\tr_1 \left( [\gamma\otimes\1,\log\gamma\otimes\1](V-V^{\gamma}\otimes\1) \right) =0
\end{align*}
for the first case. For the second case, we use $V^{\gamma}=\tr_1((\gamma\otimes\1)V)$ to arrive at 
\begin{align*}
    &\tr_1\left( \gamma\otimes\1 \left[ V-\1\otimes V^{\gamma},\1\otimes\log\gamma \right] \right)\\
    &=\left[ \tr_1 \left( (\gamma\otimes\1) V \right),\log\gamma  \right]- \left[ V^{\gamma},\log\gamma \right]\\
    &=\left[ V^{\gamma},\log\gamma \right]-\left[ V^{\gamma},\log\gamma \right]=0.   
\end{align*}
Both cases yield $\tr_{i_{\nu}}\left( \gamma X_{i_{\nu}j_{\nu}} \right)=0$. This concludes the proof.
\end{proof}

\bigskip

Next, we prove the combinatorial result in Proposition \ref{prop: combinatorics}.

\begin{proof}[Proof of Proposition \ref{prop: combinatorics}.]
    Let $[2m] = \{1, \dots, 2m\}$ be the set of positions in $(I_m, J_m)$. We represent $(I_m, J_m)$ as a mapping $\phi: [2m] \to \{1, \dots, N\}$. Conditions (2) and (3) imply that for any value $x \in \text{Im}(\phi)$, $\phi^{-1}(x)$ satisfies $|\phi^{-1}(x)| \geq 2$. 

Let $k = |\text{Im}(\phi)|$ be the number of distinct index values. The total number of ways to partition $[2m]$ into $k$ blocks such that each block has size at least 2 is given by the 2-associated Stirling number of the second kind, denoted by $S_2(2m, k)$. Accounting for the choice of $k$ values and their permutations, but ignoring condition (1) for an upper bound, we have
\begin{equation*}
    |\cI_{m,N}| \leq \sum_{k=1}^m \binom{N}{k} k! S_2(2m, k) = \sum_{k=1}^m \frac{N!}{(N-k)!} S_2(2m, k).
\end{equation*}
The exponential generating function for $S_2(n, k)$ is $\frac{(e^x - 1 - x)^k}{k!}$, i.e.
\begin{equation*}
   \sum_{n=0}^{\infty} S_2(n, k) \frac{x^n}{n!} = \frac{(e^x - 1 - x)^k}{k!}.
\end{equation*}

\textbf{Case 1:} $m\leq N$.

Using the linearity of the coefficient extraction operator $[x^n]$ and the binomial theorem, we have
\begin{align*}
   |\cI_{m,N}| &\leq \sum_{k=0}^N \frac{N!}{k!(N-k)!} (2m)! [x^{2m}] (e^x - 1 - x)^k \nonumber \\
    &= (2m)! [x^{2m}] \sum_{k=0}^N \binom{N}{k} (e^x - 1 - x)^k \nonumber \\
    &= (2m)! [x^{2m}] \left( 1 + (e^x - 1 - x) \right)^N \nonumber \\
    &= (2m)! [x^{2m}] (e^x - x)^N .
\end{align*}
Let $f(x) = (e^x - x)^N$. Since the Taylor coefficients of $e^x - x$ are non-negative, all coefficients of $f(x)$ are non-negative. For any $x > 0$, we have
\begin{equation*}
    [x^{2m}] f(x) \leq \frac{(e^x - x)^N}{x^{2m}}.
\end{equation*}
Thus, for any $x > 0$, $|\cI_{m,N}|\leq (2m)! \frac{(e^x - x)^N}{x^{2m}}$.

We choose $x = \sqrt{m/N}$. Using the elementary inequality $e^x - x \leq e^{x^2}$ for $x > 0$, we have
\begin{equation*}
    (e^x - x)^N \leq e^{Nx^2}=e^{m}.
\end{equation*}
Using the Stirling bound 
$$ \sqrt{2\pi n} \left(\frac{n}{e}\right)^{n}\leq n! \leq\frac{11}{10} \sqrt{2\pi n} \left(\frac{n}{e}\right)^{n},$$
we obtain
\begin{align*}
    |\cI_{m,N}| &\leq \frac{11}{10}\sqrt{4\pi m} \left( \frac{2m}{e} \right)^{2m} \frac{e^m}{m^m}N^m \nonumber \\
    &= \frac{11}{10}\sqrt{4\pi m} \left( \frac{4m}{e} \right)^m   N^m\\ 
    &\leq \frac{11\sqrt{2}}{10} \left( 4N \right)^m m! \\
    &\leq (C_0 N)^m m!,
\end{align*}
for instance, with the admissible choice $C_0=8$. 

\textbf{Case 2:} $m > N$.

In this case, we use a rough bound. Since each index $i_\nu$ and $j_\nu$ can take $N$ possible values, we have
\begin{equation*}
    |\cI_{m,N}| \leq N^{2m} \leq \left(\frac{N^2 e}{m}\right)^m m! \leq (e N)^m m!.
\end{equation*}

This completes the proof.
\end{proof}

\bigskip
\section{Propagation of Regularity: Proof of Proposition \ref{prop: propagation of regularity}.}\label{sec: propagation}

The goal of this section is to verify that the fluctuation operator $X(t)$ remains bounded throughout the Hartree evolution. We reduce this to a uniform commutator estimate for $[\log\gamma_t,W_{y,t}]$, where $W_{y,t}$ is the centered interaction seen from a fixed particle, and propagate the estimate by using the unitary Hartree flow together with the commutator assumptions imposed on $\gamma_0$.

\begin{proof}
    Fix $y\in\Omega$, and define $W_{y,t}(x)=V(x-y)-V^{\gamma_t}(x)$. 
    It suffices to show that
    \begin{equation*}
        \sup_{y\in\Omega}\norm{[\log\gamma_t,W_{y,t}]}_{\mathrm{op}}\leq C(t).
    \end{equation*}
    We claim that there exists a constant $\tilde{C}=C_1\norm{|\nabla V|}_{L^{\infty}}+C_2$ such that for any $f\in W^{1,\infty}(\Omega)$ and $t\geq 0$, we have
    \begin{equation}\label{eq:commutator_bound}
        \norm{[\log\gamma_t,f]}_{\mathrm{op}}\leq \norm{[\log\gamma_0,f]}_{\mathrm{op}}+2\tilde{C} \norm{f}_{L^{\infty}}t.
    \end{equation}
    This, together with the assumptions on $\gamma_0$, implies the desired result. Indeed, we have 
    \begin{align*}
        \norm{[\log\gamma_t,W_{y,t}]}_{\mathrm{op}}&\leq \norm{[\log\gamma_0,W_{y,t}]}_{\mathrm{op}}+2\tilde{C}\norm{W_{y,t}}_{L^{\infty}}t\\
        &\leq C_1\norm{\nabla W_{y,t}}_{L^{\infty}}+2\tilde{C}\norm{W_{y,t}}_{L^{\infty}}t\\
        &\leq 2C_1\norm{\nabla V}_{L^{\infty}}+4\tilde{C} \norm{V}_{L^{\infty}}t.
    \end{align*}
    Here, we used Assumption \ref{asump of gamma} in the second line and Young's inequality in the last line. It remains to show the claim \eqref{eq:commutator_bound}.
    
    The quantity $\log\gamma_t$ satisfies the corresponding commutator equation, as follows from the functional calculus. Let $U_{t,s}$ be the unitary operator solving
    \begin{align*}
        \ii\partial_t U_{t,s}&= (h+\rho_{\gamma_t}*V) U_{t,s},\\
        U_{s,s}&=\1,\quad\quad U_{t,s}^\dagger=U_{s,t}.
    \end{align*}

    We have $\gamma_t=U_{t,0}\gamma_0 U_{t,0}^\dagger$ and thus $\log\gamma_t=U_{t,0}(\log\gamma_0) U_{t,0}^\dagger$. Therefore,
    \begin{equation*}
        [\log\gamma_t,f]=U_{t,0}[\log\gamma_0,f_t]U_{t,0}^\dagger,
    \end{equation*}
    where $f_t=U_{t,0}^\dagger f U_{t,0}$, which satisfies the Heisenberg equation
    \begin{equation*}
        -\ii\partial_t f_t=[h+\rho_{\gamma_t}*V,f_t].
    \end{equation*}
    We compute that 
    \begin{align*}
        \partial_t [\log\gamma_0, f_t]&=\mathrm{i} \left[ \log\gamma_0, \left[ H(t), f_t \right]  \right] \\
        &=\mathrm{i} \left[ \left[ \log\gamma_0, H(t) \right],f_t  \right] +\mathrm{i}\left[ H(t), \left[ \log\gamma_0, f_t \right]  \right] .
    \end{align*}
    Recall that $H(t)=h+\rho_{\gamma_t}*V$. Using the Duhamel formula, we obtain
    \begin{equation*}
        [\log\gamma_0,f_t]=U_{t,0}^{\dagger}[\log\gamma_0,f]U_{t,0}+\mathrm{i}\int_0^t U_{t,s}^\dagger  \left[ \left[ \log\gamma_0, h+\rho_{\gamma_s}*V \right],f_s  \right]U_{t,s} \dd s .
    \end{equation*}
    Taking the operator norm and using the assumptions on $\gamma_0$, we conclude that
    \begin{align*}
        \norm{[\log\gamma_0,f_t]}_{\mathrm{op}}&\leq \norm{[\log\gamma_0,f]}_{\mathrm{op}}\\
        &+2C_1\int_0^t \norm{\nabla V}_{L^{\infty}}\norm{\rho_{\gamma_s}}_{L^1}\norm{f_s}_{\rm op}\dd s+2C_2\int_0^t \norm{f_s}_{\rm op} \dd s\\
        &=\norm{[\log\gamma_0,f]}_{\mathrm{op}}+2 \left( C_1\norm{|\nabla V|}_{L^{\infty}}+C_2 \right)\norm{f}_{L^{\infty}}t .
    \end{align*}
    This finishes the proof of the claim \eqref{eq:commutator_bound}.
\end{proof}

\bigskip

\section{Uniformity in the Planck constant: Proof of Theorem \ref{tm: uniform in h}.}\label{sec: uniformity}

This section proves the uniform-in-$\hbar$ estimate stated in Theorem \ref{tm: uniform in h}. The argument applies the entropy bound of Theorem \ref{tm: relative entropy bound} to the rescaled Hamiltonian with $V=\hbar^{-1}\Phi$, combines it with the quantum Wasserstein/Husimi estimate from \cite{golse2016mean}, and then optimizes the two bounds with respect to $\hbar$ to obtain a convergence rate depending only on $N$, $k$, $T$, and the potential.

\begin{proof}
   Let us treat $\hbar$ as a positive parameter. Dividing both sides of the $\hbar$-dependent von Neumann equation \eqref{eq:von-Neumann eqn-2} by $\hbar$, we can absorb $\hbar$ into the Hamiltonian and write
\begin{equation*}
    \ii\partial_t\Gamma = [H_N^\hbar,\Gamma],
\end{equation*}
where
\begin{equation*}
H_N^\hbar = \sum_{j=1}^N -\hbar \Delta_j + \frac{1}{N-1} \sum_{1\leq i<j\leq N} \frac{1}{\hbar}\Phi(x_i-x_j).
\end{equation*}
The corresponding $\hbar$-dependent Hartree equation becomes
\begin{equation*}
\ii\partial_t\gamma = \left[-\hbar \Delta + \frac{1}{\hbar}\Phi\ast\rho_{\gamma},\gamma\right].
\end{equation*}
Since Theorem \ref{tm: relative entropy bound} is quite general, we can apply it to the $\hbar$-dependent equations by choosing $h=-\hbar\Delta$ and $V=\frac{1}{\hbar}\Phi$. For $t\in[0,T]$, using the explicit bound in \eqref{eq:gronwall}, we have 
\begin{equation}\label{eq:entropy-bound-hbar}
    S(\Gamma_t^N,\gamma_t^{\otimes N})\leq e^{M_{\hbar}t}S(\Gamma_0^N,\gamma_0^{\otimes N})+(\log2)e^{M_{\hbar}t}=(\log2)e^{M_{\hbar}t},
\end{equation}
where
\begin{equation*}
    M_{\hbar} = 16C_0C_2\norm{\Phi}_{L^{\infty}}T
    +\hbar^{-1}8C_0C_1\norm{\nabla\Phi}_{L^{\infty}}
    +\hbar^{-2}16C_0C_1\norm{\nabla\Phi}_{L^{\infty}}\norm{\Phi}_{L^{\infty}}T.
\end{equation*}
Here we used the fact that $\Gamma_0^N=\gamma_0^{\otimes N}$, so the first term vanishes. 

We will use a similar proof strategy as in Section 8 of \cite{golse2018derivation}. Under the assumption that $\gamma_0$ is a Toeplitz operator, we have from \cite{golse2016mean} (using formula (17) of Theorem 2.4 and point (2) of Theorem 2.3) that
\begin{equation}\label{eq:wasserstein-bound-hbar}
    \mathrm{dist}_{\mathrm{MK},2}\left( \mathscr{H}_{\hbar}[\Gamma_t^{N:k}],\mathscr{H}_{\hbar}[\gamma_t^{\otimes k}] \right)^2\leq k \left( 2d\hbar+\frac{C}{N} \right)e^{\Lambda t} +2kd\hbar, 
\end{equation}
where 
\begin{equation*}
    \Lambda:=3+4\mathrm{Lip}(\nabla\Phi)^2\quad \text{and}\quad C:=\frac{8\norm{\nabla\Phi}_{L^{\infty}}}{\Lambda}.
\end{equation*}
The distance $\mathrm{dist}_{\mathrm{MK},2}$ is the Wasserstein distance of order 2, and $\mathscr{H}_{\hbar}[\cdot]$ is the Husimi transform defined in \eqref{eq:Husimitrans}. We recall that, for each pair of Borel probability measures on $T^*(\Omega^k)$ with finite second moments, the Wasserstein distance of order 2 is defined by
\begin{equation*}
    \mathrm{dist}_{\mathrm{MK},2}(\mu,\nu) = \left( \inf_{\pi\in\Pi(\mu,\nu)} \int_{(T^*(\Omega^k))^2} |z-z'|^2 \pi(\dd z,\dd z') \right)^{1/2}.
\end{equation*}
For fixed $\hbar>0$, \eqref{eq:entropy-bound-hbar} is a powerful bound, while the right-hand side blows up as $\hbar\to 0$. On the other hand, the bound on the Wasserstein distance \eqref{eq:wasserstein-bound-hbar} never vanishes for fixed $\hbar$, but it becomes small as $\hbar\to 0$. We will combine these two bounds.

For $t\in[0,T]$, set
\begin{equation*}
    E(N,\hbar,k,t):=\min\left\{S(\Gamma_t^{N:k},\gamma_t^{\otimes k})\quad,    \mathrm{dist}_{\mathrm{MK},2}\left( \mathscr{H}_{\hbar}[\Gamma_t^{N:k}],\mathscr{H}_{\hbar}[\gamma_t^{\otimes k}] \right)^2\right\}.
\end{equation*}
Using the block subadditivity bound in Lemma \ref{lem:block}, we have
\begin{equation*}
    E(N,\hbar,k,t)\leq\min\left\{\frac{k\log 2}{N}e^{M_{\hbar}t}\quad,2kd(e^{\Lambda t}+1)\hbar+\frac{kCe^{\Lambda t}}{N}\right\}.
\end{equation*}
For simplicity, we write $M_{\hbar}=M^{(0)}+\hbar^{-1}M^{(1)}+\hbar^{-2}M^{(2)}$, where 
\begin{equation*}
    M^{(0)}=16C_0C_2\norm{\Phi}_{L^{\infty}}T,
    \quad M^{(1)}= 8C_0C_1\norm{\nabla\Phi}_{L^{\infty}},
    \quad M^{(2)}=16C_0C_1\norm{\nabla\Phi}_{L^{\infty}}\norm{\Phi}_{L^{\infty}}T,
\end{equation*}
and define two functions:
\begin{equation*}
    f(\hbar,t):=\frac{k\log 2}{N}\exp\left(M^{(0)}t+\frac{M^{(1)}t}{\hbar}+\frac{M^{(2)}t}{\hbar^2}\right),\quad g(\hbar,t):=2kd(e^{\Lambda t}+1)\hbar+\frac{kCe^{\Lambda t}}{N}.
\end{equation*}
We have
\begin{equation*}
   \sup_{\hbar>0}\sup_{t\in[0,T]} E(N,\hbar,k,t)\leq \sup_{\hbar>0}\sup_{t\in(0,T]} \min\{f(\hbar,t),g(\hbar,t)\},
\end{equation*}
since $E(N,\hbar,k,0)=0$. 

For fixed $\hbar>0$, both $f(\hbar,t)$ and $g(\hbar,t)$ are increasing functions of $t$. On the other hand, for fixed $t\in(0,T]$, $f(\hbar,t)$ is a decreasing function of $\hbar$ while $g(\hbar,t)$ is an increasing function of $\hbar$. Therefore, for each fixed $t\in(0,T]$, there exists a unique $\hbar_t>0$ such that $f(\hbar_t,t)=g(\hbar_t,t)$. More specifically, we have
\begin{equation*}
    \sup_{\hbar>0}\sup_{t\in(0,T]}\min\{f(\hbar,t),g(\hbar,t)\}\leq f(\hbar_T,T)=g(\hbar_T,T),
\end{equation*}
where $\hbar_T$ is the unique solution to the equation $f(\hbar_T,T)=g(\hbar_T,T)$, i.e.
\begin{equation*}
    \frac{k\log 2}{N}\exp\left(M^{(0)}T+\frac{M^{(1)}T}{\hbar_T}+\frac{M^{(2)}T}{\hbar_T^2}\right)=2kd(e^{\Lambda T}+1)\hbar_T+\frac{kCe^{\Lambda T}}{N}.
\end{equation*}
This equation implies that $\hbar_T\sim(\log N)^{-1/2}$. To obtain an explicit bound, we claim that $\hbar_T<\sqrt{\frac{2M^{(2)}T}{\log N}}$ for $N$ large enough. Indeed, we can check that 
\begin{equation*}
    f\left(\sqrt{\frac{2M^{(2)}T}{\log N}},T\right)=\frac{k\log2}{\sqrt{N}}\exp \left( M^{(0)}T+M^{(1)}\sqrt{\frac{T\log N}{2M^{(2)}}} \right) ,
\end{equation*}
which is smaller than 
\begin{equation*}
    g\left(\sqrt{\frac{2M^{(2)}T}{\log N}},T\right)=2kd(e^{\Lambda T}+1)\sqrt{\frac{2M^{(2)}T}{\log N}}+\frac{kCe^{\Lambda T}}{N}
\end{equation*}
for $N$ large enough. This, together with the monotonicity of $f$ and $g$, implies that $\hbar_T<\sqrt{\frac{2M^{(2)}T}{\log N}}$. Therefore, we have
\begin{align*}
    \sup_{\hbar>0}\sup_{t\in[0,T]} E(N,\hbar,k,t)&\leq g\left(\sqrt{\frac{2M^{(2)}T}{\log N}},T\right)=2kd(e^{\Lambda T}+1)\sqrt{\frac{2M^{(2)}T}{\log N}}+\frac{kCe^{\Lambda T}}{N}\\
    &\lesssim \frac{kdT\left(\norm{\nabla\Phi}_{L^{\infty}}\norm{\Phi}_{L^{\infty}}\right)^{\frac{1}{2}}\left( e^{\Lambda T} +1\right) }{\left( \log N \right)^{\frac{1}{2}} }.
\end{align*}
Using the fact that $\mathrm{dist}_1(\mu,\nu)$ is bounded above by both $\mathrm{dist}_{\mathrm{MK},2}(\mu,\nu)$ and $\norm{\mu-\nu}_{\mathrm{TV}}$ for any pair of probability measures $\mu,\nu$ on $T^*(\Omega^k)$ (see Lemma 8.2 in \cite{golse2018derivation}), we conclude that
\begin{align*}
    &\sup_{\hbar>0}\sup_{t\in[0,T]} \mathrm{dist}_1\left( \mathscr{H}_{\hbar}[\Gamma_t^{N:k}],\mathscr{H}_{\hbar}[\gamma_t^{\otimes k}] \right)^2\\
    &\leq \sup_{\hbar>0}\sup_{t\in[0,T]}\min\left\{ \norm{\mathscr{H}_{\hbar}[\Gamma_t^{N:k}]-\mathscr{H}_{\hbar}[\gamma_t^{\otimes k}]}_{\mathrm{TV}}^2,  \mathrm{dist}_{\mathrm{MK},2}\left( \mathscr{H}_{\hbar}[\Gamma_t^{N:k}],\mathscr{H}_{\hbar}[\gamma_t^{\otimes k}] \right)^2\right\}.
\end{align*}
Moreover, the resolution of the identity \eqref{eq:resolution of identity} implies that $\norm{\mathscr{H}_{\hbar}[\Gamma_t^{N:k}]-\mathscr{H}_{\hbar}[\gamma_t^{\otimes k}]}_{\mathrm{TV}}$ is bounded by $\norm{\Gamma_t^{N:k}-\gamma_t^{\otimes k}}_{\mathrm{tr}}$, which is in turn bounded by $\sqrt{2S(\Gamma_t^{N:k},\gamma_t^{\otimes k})}$ by quantum Pinsker's inequality in Lemma \ref{lem:pinsker}. Therefore, we have
\begin{align*}
    &\sup_{\hbar>0}\sup_{t\in[0,T]} \mathrm{dist}_1\left( \mathscr{H}_{\hbar}[\Gamma_t^{N:k}],\mathscr{H}_{\hbar}[\gamma_t^{\otimes k}] \right)^2\lesssim \sup_{\hbar>0}\sup_{t\in[0,T]}E(N,\hbar,k,t)\\
    &\lesssim \frac{kdT\left(\norm{\nabla\Phi}_{L^{\infty}}\norm{\Phi}_{L^{\infty}}\right)^{\frac{1}{2}}\left( e^{\Lambda T} +1\right) }{\left( \log N \right)^{\frac{1}{2}} }\lesssim \frac{kdT\norm{\Phi}_{W^{1,\infty}}\left( e^{\Lambda T} +1\right) }{\left( \log N \right)^{\frac{1}{2}} }.
\end{align*}
This completes the proof.
\end{proof}

\bigskip

\section{Open Systems and Lindblad Framework: Proof of Theorem \ref{thm: finite main}.}\label{sec: Lindblad}

We follow the same strategy as in the proof of Theorem \ref{tm: relative entropy bound} given in Section \ref{sec: relative entropy proof}. The main difference is that we need to handle the dissipative part of the Lindblad equation. We need only modify Propositions \ref{prop: evolution entropy} and \ref{prop: propagation of regularity}.

Let $\delta H_N(t)$ be the Hamiltonian defect defined by
\begin{equation*}
    \delta H_N(t):=\frac{1}{N-1}\sum_{1\leq i<j\leq N} W_{ij}-\sum_{j=1}^N V^{\gamma_t}_j.
\end{equation*}
We have the following result, which is similar to Proposition \ref{prop: evolution entropy}.

\begin{proposition}\label{prop: finite entropy production}
Let $\Gamma_t^N$ solve \eqref{eq:finite-Nbody-Lindblad}, and let $\gamma_t$ solve \eqref{eq:finite-Hartree-Lindblad}. Then
\begin{equation}\label{eq:finite-entropy-production}
    \frac{\dd}{\dd t}S(\Gamma_t^N,\gamma_t^{\otimes N})
    \leq
    \tr_{\gH_N}\big(\Gamma_t^N A_t^N\big),
\end{equation}
where $\sigma_t:=\gamma_t^{\otimes N}$ and
\begin{equation*}
    A_t^N:=-\ii[\delta H_N(t),\log\sigma_t].
\end{equation*}
Moreover, we can rewrite $A_t^N$ as
\begin{equation*}
    A_t^N=\frac{1}{N-1}\sum_{1\leq i\neq j\leq N}X_{ij}(t),
\end{equation*}
where the two-body fluctuation operator $X(t)$ on $\gH\otimes\gH$ is defined by
\begin{equation}\label{eq:finite-def-X}
    X(t):=-\ii \left[ W-V^{\gamma_t}\otimes\1, \ (\log\gamma_t)\otimes\1 \right] .
\end{equation}
\end{proposition}

\begin{proof}
    We first prove the entropy production estimate. The decomposition of $A_t^N$ follows by the same algebraic computation as in the discussion following Proposition \ref{prop: evolution entropy}, with $V(x_i-x_j)$ replaced by $W_{ij}$ and $V\ast\rho_{\gamma_t}$ replaced by $V^{\gamma_t}$.

    \medskip
    \noindent\textbf{Step 1.} Define $\sigma_t=\gamma_t^{\otimes N}$. We claim that
    \begin{equation*}
    \partial_t \sigma_t
    =
    -\ii[H_N^{\rm mf}(t),\sigma_t]
    +
    \sum_{j=1}^N \cL_{L_j}(\sigma_t),
\end{equation*}
where
\begin{equation*}
    H_N^{\rm mf}(t):=\sum_{j=1}^N \big(h_j+(V^{\gamma_t})_j\big).
\end{equation*}
Indeed, by the product rule,
\[
\partial_t \sigma_t
=
\sum_{j=1}^N
\gamma_t^{\otimes(j-1)}\otimes(\partial_t\gamma_t)\otimes\gamma_t^{\otimes(N-j)}.
\]
Using \eqref{eq:finite-Hartree-Lindblad},
\[
\partial_t \sigma_t
=
\sum_{j=1}^N
\gamma_t^{\otimes(j-1)}
\otimes
\Big(-\ii[h+V^{\gamma_t},\gamma_t]+\cL_L(\gamma_t)\Big)
\otimes
\gamma_t^{\otimes(N-j)}.
\]
The Hamiltonian part is
\[
-\ii[H_N^{\rm mf}(t),\sigma_t],
\]
while the dissipative part is
\[
\sum_{j=1}^N \cL_{L_j}(\sigma_t).
\]
This proves the claim.

\medskip
\noindent\textbf{Step 2.} Define the common generator
\begin{equation*}
    \widetilde{\cL}_{N,t}(A)
    :=
    -\ii[H_N^{\rm mf}(t),A]+\sum_{j=1}^N \cL_{L_j}(A).
\end{equation*}
Then
\[
\partial_t \sigma_t=\widetilde{\cL}_{N,t}(\sigma_t),
\]
while
\begin{equation*}
    \partial_t\Gamma_t^N
    =
    \widetilde{\cL}_{N,t}(\Gamma_t^N)-\ii[\delta H_N(t),\Gamma_t^N].
\end{equation*}

\begin{lemma}\label{lem: finite entropy comparison}
Let $t\mapsto \eta_t$ and $t\mapsto \zeta_t$ be sufficiently regular families of faithful density operators, and assume
\begin{align*}
    \partial_t\eta_t &= \widetilde{\cL}_t(\eta_t)+K_t(\eta_t), \\
    \partial_t\zeta_t &= \widetilde{\cL}_t(\zeta_t), 
\end{align*}
where $\widetilde{\cL}_t$ generates a completely positive trace-preserving flow, and where
\[
K_t(A)=-\ii[B_t,A]
\]
for some bounded self-adjoint operator $B_t$. Then
\begin{equation*}
    \frac{\dd}{\dd t}S(\eta_t,\zeta_t)
    \leq
    \tr\big(\eta_t(-\ii[B_t,\log\zeta_t])\big).
\end{equation*}
\end{lemma}
Apply Lemma \ref{lem: finite entropy comparison} with
\[
\eta_t=\Gamma_t^N,\qquad \zeta_t=\sigma_t,\qquad B_t=\delta H_N(t),
\]
and with common generator $\widetilde{\cL}_{N,t}$. This yields \eqref{eq:finite-entropy-production}.
\end{proof}

\medskip
\begin{proof}[Proof of Lemma \ref{lem: finite entropy comparison}.]
    Fix $t$ and let $\Phi_{t+h,t}$ be the completely positive trace-preserving propagator associated with
\[
\partial_s X_s=\widetilde{\cL}_s(X_s),\qquad s\in [t,t+h].
\]
Then
\[
\zeta_{t+h}=\Phi_{t+h,t}(\zeta_t),
\]
and
\[
\eta_{t+h}
=
\Phi_{t+h,t}\big(\eta_t+hK_t(\eta_t)\big)+o(h)
\]
in trace norm.
By monotonicity of relative entropy under completely positive trace-preserving maps,
\[
S(\eta_{t+h},\zeta_{t+h})
\leq
S(\eta_t+hK_t(\eta_t),\zeta_t)+o(h).
\]
Subtracting $S(\eta_t,\zeta_t)$, dividing by $h$, and letting $h\to 0^+$, Lemma \ref{lem:RE-derivative} gives
\[
\frac{\dd}{\dd t}S(\eta_t,\zeta_t)
\leq
\tr\big(K_t(\eta_t)(\log\eta_t-\log\zeta_t)\big).
\]
Now $K_t(\eta_t)=-\ii[B_t,\eta_t]$, and
\[
\tr\big(-\ii[B_t,\eta_t]\log\eta_t\big)=0
\]
because $\eta_t$ commutes with $\log\eta_t$. Therefore
\[
\frac{\dd}{\dd t}S(\eta_t,\zeta_t)
\leq
-\tr\big(\ii[B_t,\eta_t]\log\zeta_t\big)
=
\tr\big(\eta_t(-\ii[B_t,\log\zeta_t])\big).
\]
\end{proof}

Next, we propagate the regularity of the two-body fluctuation operator $X(t)$. The following result is analogous to Proposition \ref{prop: propagation of regularity}.

\begin{proposition}\label{prop: finite X bound}
Let $X(t)$ be defined by \eqref{eq:finite-def-X}. Then
\begin{equation*}
     \sup_{t\geq0}\norm{X(t)}_{\rm op}
    \leq
    C_W^{(m_0)}
    \leq
    \frac{4\norm{W}_{\rm op}}{m_0},
\end{equation*}
where $C_W^{(m_0)}$ is defined by \eqref{eq:finite-def-CW-general}.
\end{proposition}

\begin{proof}
The first inequality follows directly from the definition of $C_W^{(m_0)}$ and from the uniform lower bound $\gamma_t\geq m_0\1$, as noted in Remark \ref{rmk:finite-uniform-lower-bound-gammaT}. It remains to prove the explicit estimate in terms of $\norm{W}_{\rm op}$.

Using Lemma \ref{lem:comm-log} with $X=\gamma_t\otimes\1$ and $B=W-V^{\gamma_t}\otimes\1$, we get
\begin{align*}
\norm{X(t)}_{\rm op}
&=
\norm{
-\ii\left[
W-V^{\gamma_t}\otimes\1,\,
(\log\gamma_t)\otimes\1
\right]
}_{\rm op}\\
& \leq
\int_0^\infty
\norm{(\gamma_t+s\1)^{-1}}_{\rm op}^2
\,
\norm{
[\gamma_t\otimes\1,\,
W-V^{\gamma_t}\otimes\1]
}_{\rm op}
\,\dd s \\
&\leq
\frac{1}{m_0}
\norm{
[\gamma_t\otimes\1,\,
W-V^{\gamma_t}\otimes\1]
}_{\rm op}.
\end{align*}
Since $\norm{\gamma_t}_{\rm op}\leq1$ and
$\norm{V^{\gamma_t}}_{\rm op}
=
\norm{\tr_2((\1\otimes\gamma_t)W)}_{\rm op}
\leq
\norm{W}_{\rm op}$, we have
\[
\norm{
[\gamma_t\otimes\1,\,
W-V^{\gamma_t}\otimes\1]
}_{\rm op}
\leq
2\norm{\gamma_t\otimes\1}_{\rm op}
\norm{W-V^{\gamma_t}\otimes\1}_{\rm op}
\leq
4\norm{W}_{\rm op}.
\]
Combining the two estimates gives
\[
    \norm{X(t)}_{\rm op}
    \leq
    \frac{4\norm{W}_{\rm op}}{m_0}.
\]

This proves the claim.

\end{proof}

\bigskip

Now we are ready to prove Theorem \ref{thm: finite main}.

\begin{proof}[Proof of Theorem \ref{thm: finite main}.]
By Proposition \ref{prop: finite entropy production}, we have
\[
    \frac{\dd}{\dd t}S(\Gamma_t^N,\gamma_t^{\otimes N})
    \leq
    \tr_{\gH_N}(\Gamma_t^N A_t^N),
    \qquad
    A_t^N=\frac{1}{N-1}\sum_{1\leq i\neq j\leq N}X_{ij}(t).
\]
The cancellation rule used in Proposition \ref{prop:cancellation} remains valid in the present finite-dimensional setting. Indeed, if the first tensor factor of one copy of $X$ appears only once, then
\[
\tr_1\left(
(\gamma\otimes\1)
\left[
W-V^\gamma\otimes\1,\,
(\log\gamma)\otimes\1
\right]
\right)
=
\tr_1\left(
[\gamma\otimes\1,(\log\gamma)\otimes\1]
(W-V^\gamma\otimes\1)
\right)
=0.
\]
If the second tensor factor appears only once, then using the exchange symmetry of $W$ and the definition of $V^\gamma$ gives
\[
\tr_2\left(
(\1\otimes\gamma)
\left[
W-\1\otimes V^\gamma,\,
\1\otimes\log\gamma
\right]
\right)
=
[V^\gamma,\log\gamma]-[V^\gamma,\log\gamma]
=0.
\]
Thus the same mixed-moment cancellation and the same combinatorial estimate as in Proposition \ref{prop: combinatorics} apply. Consequently, for every $m\geq0$,
\[
    \left|
    \tr\left(\gamma_t^{\otimes N}(A_t^N)^m\right)
    \right|
    \leq
    \left(2C_0\norm{X(t)}_{\rm op}\right)^m m!.
\]
The entropy inequality Lemma \ref{lem:entropy ineq}, with
$\lambda=(4C_0\norm{X(t)}_{\rm op})^{-1}$, then gives
\[
    \frac{\dd}{\dd t}S(\Gamma_t^N,\gamma_t^{\otimes N})
    \leq
    4C_0\norm{X(t)}_{\rm op}
    S(\Gamma_t^N,\gamma_t^{\otimes N})
    +
    4C_0(\log2)\norm{X(t)}_{\rm op}.
\]
Using Proposition \ref{prop: finite X bound} and the admissible choice $C_0=8$ from Proposition \ref{prop: combinatorics}, we obtain 
\[
    \frac{\dd}{\dd t}S(\Gamma_t^N,\gamma_t^{\otimes N})
    \leq
    32 C_W^{(m_0)}
    S(\Gamma_t^N,\gamma_t^{\otimes N})
    +
    32(\log2)C_W^{(m_0)}.
\]
A Gronwall argument yields
\[
    S(\Gamma_t^N,\gamma_t^{\otimes N})
    \leq
    e^{32 C_W^{(m_0)}t}
    \Big(S(\Gamma_0^N,\gamma_0^{\otimes N})+\log2\Big).
\]

Finally, using $C_W^{(m_0)}\leq 4\norm{W}_{\rm op}/m_0$ gives the explicit bound stated in the theorem.
\end{proof}


\end{document}